\documentclass[10pt, letterpaper]{article}
\usepackage[top = 0.85in, left = 2.25in, footskip = 0.75in]{geometry} 
\usepackage[shortlabels]{enumitem} 
\usepackage{amsmath, amsthm, amssymb}
\usepackage{mathtools} 
\usepackage{changepage}

\usepackage{textcomp, marvosym}

\usepackage{cite}

\usepackage{nameref, hyperref}

\usepackage[right]{lineno}

\usepackage[nopatch = eqnum]{microtype}
\DisableLigatures[f]{encoding = *, family = *}

\usepackage[table]{xcolor}

\usepackage{array}

\newcolumntype{+}{!{\vrule width 2pt}}

\newlength\savedwidth

\usepackage{ragged2e} 
\justifying 
\usepackage[aboveskip = 2pt, labelfont = bf, labelsep = period, singlelinecheck = off]{caption}

\makeatletter
\renewcommand{\@biblabel}[1]{\quad#1.}
\makeatother

\usepackage{lastpage, fancyhdr, graphicx}
\usepackage{epstopdf}
\fancyheadoffset[L]{2.25in}
\fancyfootoffset[L]{1.75in} 
\begin{document}
\vspace*{0.125in}


\begin{flushleft}

{\Large
    \textbf\newline{Comparing AI versus Optimization Workflows for Simulation-Based Inference of Spatial-Stochastic Systems}
}
\newline
\\
Michael A. Ramirez-Sierra\textsuperscript{1,2},
Thomas R. Sokolowski\textsuperscript{1*}
\\
\bigskip
\textbf{1} Frankfurt Institute for Advanced Studies (FIAS), Ruth-Moufang-Straße 1, 60438 Frankfurt am Main, Germany
\\
\textbf{2} Goethe-Universität Frankfurt am Main, Faculty of Computer Science and Mathematics, Robert-Mayer-Straße 10, 60054 Frankfurt am Main, Germany
\\
\bigskip

* ramirez-sierra@fias.uni-frankfurt.de

\end{flushleft}



\section*{Abstract}

Model parameter inference is a universal problem across science. This challenge is particularly pronounced in developmental biology, where faithful mechanistic descriptions require spatial-stochastic models with numerous parameters, yet quantitative empirical data often lack sufficient granularity due to experimental limitations. Parameterizing such complex models therefore necessitates methods that elaborate on classical Bayesian inference by incorporating notions of optimality and goal-orientation through low-dimensional objective functions that quantitatively capture the target behavior of the underlying system. In this study, we contrast two such inference workflows and apply them to biophysics-inspired spatial-stochastic models. Technically, both workflows are simulation-based inference (SBI) methods. The first method leverages a modern deep-learning technique known as sequential neural posterior estimation (SNPE), while the second is based on a classical optimization technique called simulated annealing (SA). We evaluate these workflows by inferring the parameters of two complementary models for the inner cell mass (ICM) lineage differentiation in the blastocyst-stage mouse embryo. This developmental biology system serves as a paradigmatic example of a highly robust and reproducible cell-fate proportioning process that self-organizes under strongly stochastic conditions, such as intrinsic biochemical noise and cell-cell signaling delays. Our results indicate that while both methods largely agree in their predictions, the modern SBI workflow provides substantially richer inferred distributions at an equivalent computational cost. We identify the computational scenarios that favor the modern SBI method over its classical counterpart. Finally, we propose a plausible approach to integrate these two methods, thereby synergistically exploiting their parameter space exploration capabilities.


\section*{Author summary}

Mechanistic models provide an in-depth understanding of important biophysical systems. In fields such as developmental biology, these models are inherently complex, as they require genuinely spatial-stochastic descriptions of the underlying systems. This complexity poses significant challenges for inferring model parameters. Recently, modern deep-learning techniques have been integrated with simulation-based inference, creating an exciting new approach for estimating parameters of such models. Their overall goal is to computationally replicate target empirical observations and to develop powerful prediction tools for uncovering hidden system dynamics. However, these modern approaches remain broadly general and are often difficult to implement for specific spatial-stochastic problems, particularly within developmental biology. This difficulty raises the question of how much more valuable these modern approaches are compared to classical techniques. In this study, we compare one modern approach, \emph{AI-MAPE}, inspired by the sequential neural posterior estimation (SNPE) algorithm, against one classical approach, \emph{SA-SGM}, inspired by the simulated annealing (SA) algorithm. Our findings show that, while the inferred parameter sets generally agree between the two approaches, the AI-powered method, at comparable computational effort, provides significantly richer and more regular inferred distributions. This results in more detailed information about parameter interactions and synergies than the SA-inspired method.


\section*{Introduction}

In contrast to phenomenological approaches, biophysics-inspired mechanistic models represent the dynamics of biological systems with higher fidelity, albeit at the expense of more complex modeling techniques. In fact, an in-depth understanding of biological problems requires modeling approaches that not only match existing empirical data but also provide quantifiable and testable predictions beyond the initial scope of experimental observations \cite{torregrosa_mechanistic_2021}.

In developmental biology, where system dynamics often exhibit strong robustness and reproducibility despite highly variable conditions, spatial-stochastic descriptions are essential \cite{mjolsness_prospects_2019}. These descriptions necessitate access to rich experimental data. Recent advances in open-source big data tools and ubiquitous computational power have enabled the incorporation of vast amounts of data into computational biology models \cite{bonnaffoux_wasabi_2019, cang_multiscale_2021, forsyth_iven_2021, jiang_identification_2022, dirk_recognition_2023, alamoudi_fitmulticell_2023, prescott_efficient_2024}. However, high-dimensional stochastic models require not only large and comprehensive datasets but also fine-grained, spatial multi-scale time series information for in-depth mechanistic approaches \cite{verdier_simulation-based_2023, wang_missing_2024}. Despite these advancements, detailed quantitative data remain scarce for many biological systems, leading to a predominance of qualitative descriptions.

Given these conditions, model parameter inference arguably becomes the most challenging step in constructing mechanistic, spatial-stochastic models for developmental biophysics problems \cite{wang_massive_2019, coulier_multiscale_2021, sukys_approximating_2022, coulier_systematic_2022, raimundez_posterior_2023}. Classical approaches such as approximate Bayesian computation (ABC), simulated annealing (SA), and heuristic tuning have been widely used \cite{schnoerr_approximation_2017, cranmer_frontier_2020, jiang_identification_2022, franzin_landscape-based_2023, alamoudi_fitmulticell_2023, stillman_generative_2023, prescott_efficient_2024}. However, these methods have limited capabilities in handling high-dimensional parameter spaces.

To address this challenge, biologically-informed neural networks (BINNs) and deep-learning simulation-based inference (DL-SBI) workflows have become increasingly relevant \cite{cranmer_frontier_2020, lagergren_biologically-informed_2020, perez_efficient_2022, wang_missing_2024}. Recent DL-SBI algorithms, such as the sequential neural posterior estimation (SNPE) variant C, are particularly favorable for performing inference in likelihood-free problems by combining the flexibility and generalization power of artificial neural networks (ANNs) with autoregressive models, also known as masked autoregressive flow (MAF) techniques \cite{papamakarios_masked_2018, cranmer_frontier_2020}. The SNPE algorithm allows for non-amortized inference over several rounds, focusing on a single target observation (see Methods section \nameref{subsection:simulation_pipeline} for additional details), and significantly improves computational efficiency over its amortized counterpart \cite{greenberg_automatic_2019, deistler_truncated_2022, boelts_simulation-based_2023, xiong_efficient_2023, dirmeier_simulation-based_2024}.

While DL-SBI approaches have been influential in various fields such as neuroscience \cite{goncalves_training_2020, kaiser_simulation-based_2023}, astrophysics \cite{kolmus_tuning_2024}, and structural biology \cite{dingeldein_simulation-based_2023, dingeldein_amortized_2024}, their application in developmental biology has been limited. Nonetheless, their potential is already recognized, particularly in problems with limited quantitative data such as morphogenesis \cite{stillman_generative_2023, ramirez-sierra_ai-powered_2024}. These inference workflows are evolving rapidly, but their applicability remains challenging \cite{tejero-cantero_sbi_2020, gorecki_amortized_2023}. Few studies so far have addressed the advantages and drawbacks of DL-SBI methods compared to classical methods and whether it is worth the effort to adopt these novel approaches.

In this study, we compare two inference workflows. The first, AI-MAPE, is inspired by the SNPE algorithm \cite{greenberg_automatic_2019} and uses maximum-a-posteriori estimation (MAPE) for model parameter selection \cite{lueckmann_benchmarking_2021}. The second, SA-SGM, is inspired by the simulated annealing (SA) algorithm \cite{franzin_landscape-based_2023} and incorporates elements typical of reinforcement learning, using a self-devised methodology inspired by the sample geometric median (SGM) for model parameter selection \cite{minsker_geometric_2015, minsker_robust_2016}. Both methods are simulation-based inference (SBI) workflows, as detailed in the \nameref{section:methods} section.

We explore the conditions under which AI-MAPE is superior to SA-SGM and how it can enhance mechanistic modeling approaches for spatial-stochastic, biophysical systems. In such models, identifying parameter sets or distributions that drive simulated dynamics to match empirical observations and typical patterns is essential \cite{invernizzi_skipping_2022, stillman_generative_2023}. In particular, these parameter sets should replicate target system behavior which emerges despite biochemical noise and cell-cell signaling delays, reflecting the robustness and reproducibility of biological systems \cite{tolley_methods_2024}. Moreover, mechanistic models help uncover non-obvious relations and synergies among system variables, playing an integral role in biological research by disentangling causal relationships between stimuli and responses. These relationships dictate the temporal evolution of living cells and coordinate individual functional roles at higher levels of biological organization \cite{sgro_intracellular_2015, kaiser_simulation-based_2023}.

To benchmark their effectiveness, we apply both methods to perform parameter inference for two complementary spatial-stochastic gene regulation models of the developing mouse blastocyst \cite{chowdhary_journey_2022}. This paradigmatic biological system exemplifies a self-organizing and reproducible process with robust cell-lineage proportioning and reliable cell-sorting \cite{bruckner_information_2024}. Our focus is on the subsystem controlling the proportions of cell-fate subpopulations upon differentiation, a subject of our companion study \cite{ramirez-sierra_ai-powered_2024}, which explores the implications of stochasticity (or noise) and spatial inhomogeneities for correct gene expression patterns during early mouse embryo development. Thus, our simulation framework for studying cell-type proportioning under genuinely spatial-stochastic conditions forms the basis for our technical cross-comparison of the inference methods described here and their applicability to developmental biology systems.

Finally, we propose an approach to integrate both methods, potentially enhancing their advantages (powerful parameter-space interpolation and intra-round biased sampling) \cite{franzin_revisiting_2019, boelts_flexible_2022, boelts_simulation-based_2023}, while mitigating common pitfalls (over-confident predictions and limited generalizability of parameter estimates) \cite{albert_simulated_2015, frank_input-output_2013, kaiser_simulation-based_2023, xiong_efficient_2023, dirmeier_simulation-based_2024}.


\section*{Methods} \label{section:methods}

\graphicspath{{./Manuscript_Graphics/}} 


In this section, we provide a detailed summary of the two inference workflows (AI-MAPE versus SA-SGM) employed for this study. Common aspects between these two workflows are presented first (see also Fig~\ref{fig1}[A-D]), while their distinctive and particular attributes are explained separately later (see also Fig~\ref{fig1}[AI-MAPE, SA-SGM]). Only a concise description of the underlying biological problem and its nonlinear-multiscale model representation is provided here, limiting the presented details to a necessary minimum. Further interested readers are referred to our companion study \cite{ramirez-sierra_ai-powered_2024} for a comprehensive treatment of the biological characteristics and the applied computational modeling methodology of the considered system.

\subsection*{Mouse ICM-lineage differentiation relies on robust and reproducible stem-cell fate proportioning} \label{subsection:stem_cell_fate_proportioning}

The developmental biology problem under study is a paradigmatic example of self-organizing systems: the blastocyst-stage mouse embryo exhibits a robust inner cell mass (ICM) fate differentiation program establishing reproducible fate proportions and operating in tandem with a reliable cell-sorting mechanism, without any maternal cues \cite{chowdhary_journey_2022, bruckner_information_2024}.

For this particular work, we integrate only the main gene regulatory processes controlling cell-lineage specification and proportioning for the mouse ICM-derived progenies (epiblast ``EPI'' and primitive endoderm ``PRE'') under spatially resolved conditions, bypassing any description of mechanical interactions among cells (cellular division, proliferation, or cell motility), which act as extrinsic noise factors.

\begin{adjustwidth}{-1.75in}{0in} 
\includegraphics[width = 7in, height = 8.75in]{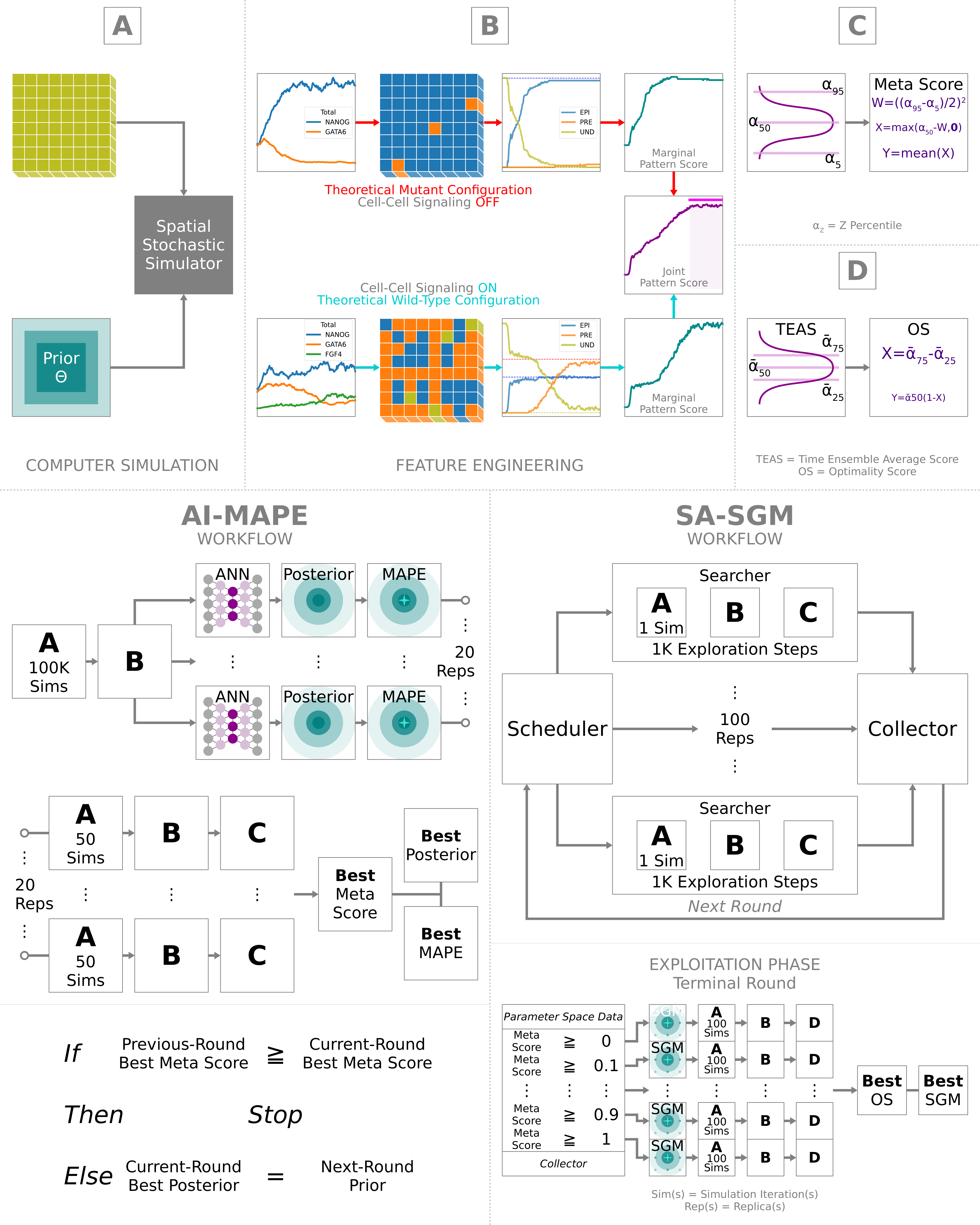} \centering 
\end{adjustwidth} 
\begin{figure}[hpt!]
\begin{adjustwidth}{-1.75in}{0in} 
\caption{{\bf Visual summary of workflows: AI-MAPE versus SA-SGM.} {\bf [A-D]} Common aspects between both workflows. {\bf [A, B]} Simulation data generation and processing pipelines. Computer simulation involves a representation of the biological system (developing ICM) as a static 2D laminar cell lattice of biochemical reaction volumes coupled via a diffusive signaling molecule (FGF4), and the prescription of a prior distribution over the model parameter space, which together feed the spatial-stochastic simulator. Feature engineering encompasses a series of data transformations: reshaping the simulation trajectories onto a regular time mesh; determining system observables at cell scale from relevant protein count time series (total NANOG, GATA6, and/or FGF4 levels); deciding cell-fate classification at tissue scale based on particular protein count thresholds (\emph{EPI lineage} NANOG $\ge$ 388 and GATA6 $\le$ 329, \emph{PRE lineage} NANOG $\le$ 329 and GATA6 $\ge$ 842, \emph{UND (undifferentiated) lineage} otherwise); constructing pattern score time series ensembles. {\bf [C]} Meta Score. We condense into a sole scalar the success of each sampled parameter vector, with respect to the target system behavior, calculating it on an ensemble of pattern score time series. This serves as a stopping criterion for AI-MAPE, and an exploration-guiding criterion for SA-SGM. {\bf [D]} TEAS and OS. We characterize accuracy and precision of the target developmental behavior by calculating these sample statistics on empirical pattern-score distributions. They also allow us to compare results between both workflows, facilitating the selection of the best or optimal parameter sets. See Methods section for complete details. {\bf [AI-MAPE]} During the simulation data generation, processing, and feature engineering phase, we produce 100 thousand spatial-stochastic trajectories, creating a single pattern-score time series ensemble. During the inference phase, this ensemble of pattern-score time series trains 20 separate ANN replicas, allowing us to obtain several posterior distribution approximations. For each posterior, we calculate its MAPE and from it we generate 50 \emph{de novo} simulations, gathering a fresh pattern-score time series ensemble. To select the optimal MAPE parameter set, we pick the empirical pattern-score distribution with the best (highest) associated meta score. If the previous-round best meta score is larger than or equal to the current-round best meta score, then we stop the workflow; otherwise the current-round best posterior is assigned as the next-round prior (or proposal), iterating the procedure. {\bf [SA-SGM]} During the exploration phase, an annealing scheduler coordinates 100 separate searcher replicas, guiding their first round and first step of parameter space probing by using the prescribed prior distribution. Searchers continue until exhausting 1000 exploring steps, transferring their data into an annealing collector; this collector informs the scheduler about starting locations for the next round of parameter space search. During the exploitation phase, the parameter space dataset (model parameter samples) from all searchers and all rounds is partitioned into hierarchically increasing clusters based on particular meta score thresholds. For each cluster, we calculate its SGM and from it we generate 100 \emph{de novo} simulations, gathering a fresh pattern-score time series ensemble. To select the optimal SGM parameter set, we pick the empirical pattern-score distribution with the best (highest) associated OS.}
\label{fig1}
\end{adjustwidth} 
\end{figure}

Focusing on the specification-proportioning process of EPI and PRE lineages from the ICM progenitor population, we constructed a biophysics-inspired, stochastic-mechanistic description of its (cell scale) gene regulation network and its (tissue scale) diffusion-based communication. The main drivers of the ICM differentiation process are the self-activation of \textit{Nanog} and \textit{Gata6} genes (primary markers of EPI and PRE lineages, respectively), together with their mutual repression \cite{saiz_coordination_2020}. Another fundamental driver of this process is an FGF4-mediated feedback loop, which enables cell-cell communication to control the associated cellular fate proportioning \cite{plusa_common_2020}.

In our modeling approach, cell-scale interactions cover the essential processes driving gene expression dynamics: binding and unbinding of transcription-factor-promoter complexes; synthesis and degradation of mRNA and protein molecules; activation and inactivation of enzymes and receptors; phosphorylation and dephosphorylation of protein complexes. At the tissue scale, interactions cover three distinct cellular signaling modes: autocrine, paracrine, and juxtacrine-like (membrane-level exchange of ligand molecules) communication. For tractability, the developing ICM tissue is represented as a static, two-dimensional (2D) spatial lattice of biochemical reaction voxels or volumes (serving as individual cells), which are coupled via FGF4 signaling, mimicking a monolayer or 2D cellular culture (see also Fig~\ref{fig2} for additional details).

Our companion study \cite{ramirez-sierra_ai-powered_2024} delves deeply into the biological connotations of the parameter interactions uncovered via the AI-MAPE method treated here. These parameter interactions provide mechanistic insights into the realization and maintenance of the target cell-fate proportions for the underlying developmental system.

\begin{adjustwidth}{-1.75in}{0in} 
\includegraphics[width = 4.5in, height = 6.75in]{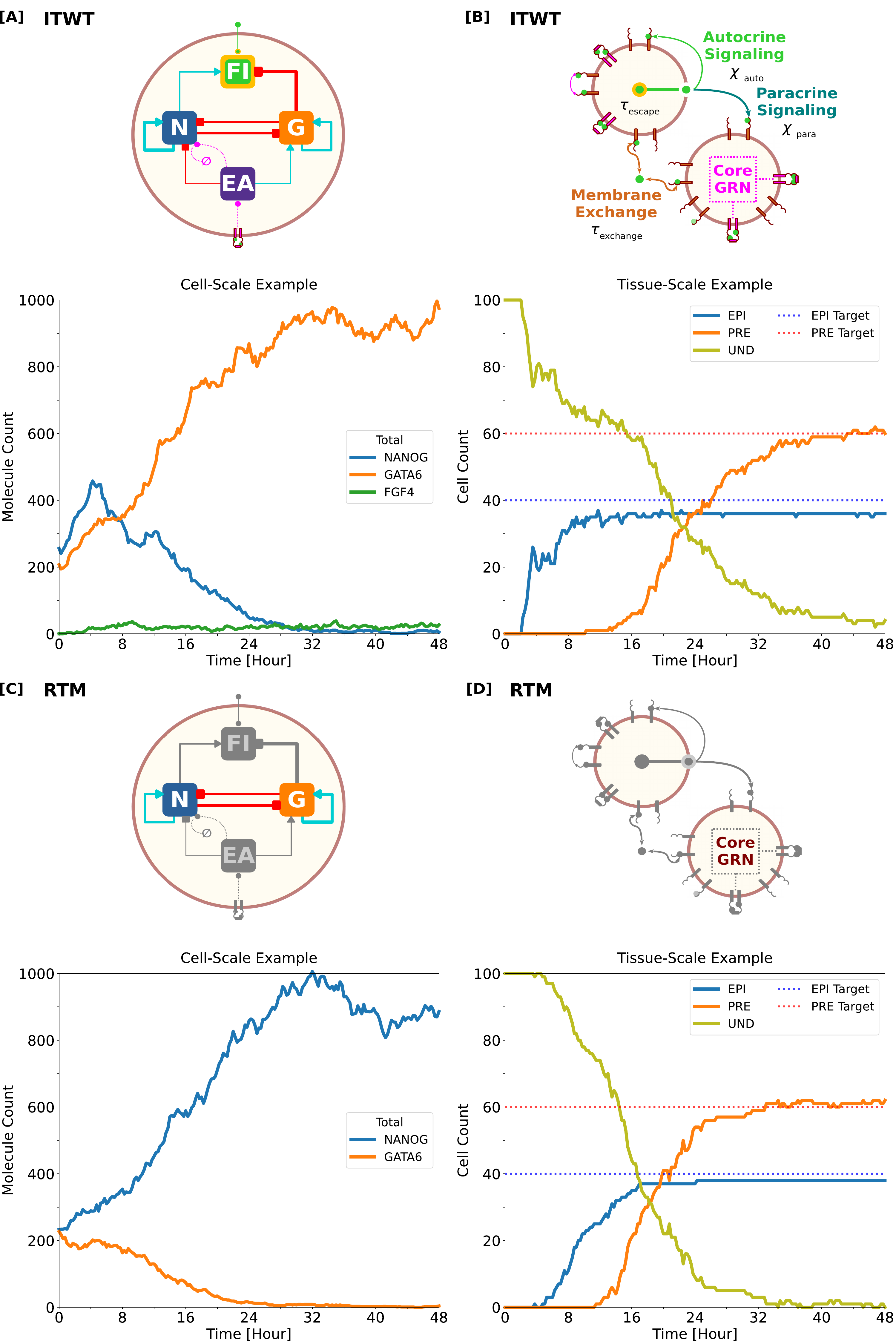} \centering 
\end{adjustwidth} 
\begin{figure}[hpt!]
\begin{adjustwidth}{-1.75in}{0in} 
\caption{{\bf ITWT and RTM systems: cell- and tissue-scale stochastic trajectories.} {\bf [A, B]} The ITWT is a wild-type-like model of cell-fate proportioning in the mouse embryo inner cell mass (ICM). {\bf [C, D]} The RTM is a mutant-like model lacking FGF4 signaling (no intercellular communication). {\bf [A, C]} Core GRN interactions (top rows); example stochastic trajectories of intracellular or cell-scale dynamics (bottom rows); only main proteins are shown for each system (NANOG, GATA6, and FGF4). {\bf [B, D]} Basic signaling-topology features (top rows); there are three distinct cell-cell communication modes (autocrine signaling, paracrine signaling, and membrane-level exchange of ligand molecules); example stochastic trajectories of extracellular or tissue-scale dynamics (bottom rows); there are three possible cell fates (Epiblast ``EPI'', Primitive Endoderm ``PRE'', and Undifferentiated ``UND''). Notice the disabled components (gray color) for the RTM compared with the ITWT.}
\label{fig2}
\end{adjustwidth} 
\end{figure}

\subsection*{Simulation data generation and processing pipelines} \label{subsection:simulation_pipeline}

Fig~\ref{fig1} illustrates the key stages and general considerations of our exploratory Bayesian inferential frameworks, with which we carry out parameter inference for two separate models. Both models aim to recapitulate the final ratio (cellular counts) between the two ICM progeny lineages (EPI and PRE) within a specific time window, as identified in the fully formed mouse blastocyst, but they differ in how they achieve this outcome. The Reinferred Theoretical Mutant (RTM) is a toy system lacking cell-cell communication, and therefore attains the ICM cell-fate ratio in a completely cell-autonomous fashion. Conversely, the Inferred Theoretical Wild Type (ITWT), which represents the actual biologically relevant system, crucially relies on cell-cell communication via the signaling molecule FGF4 and integrates all the key empirical observations of its behavior.

\subsubsection*{Prior distribution construction}

The chosen priors are multivariate uniform distributions with four (RTM) and nineteen (ITWT) dimensions or vector components. Each vector component has an associated predefined value range; these ranges are set based on literature findings and typical values known from similar systems. See S~Table~\ref{S1_Table} summarizing model parameter definitions and explanations.

\subsubsection*{Simulation data generation}

To generate spatial-stochastic trajectories of the system dynamics, we constructed an event-driven simulator based on the Next Subvolume Method (NSM) \cite{fange_stochastic_2010, erban_stochastic_2020}, which exploits the Reaction-Diffusion Master Equation (RDME) formalism \cite{barrows_parameter_2023}. The RDME methodology is a mesoscale stochastic modeling approach for spatially resolved biochemical systems, faithfully incorporating intrinsic noise and allowing for exact simulation of temporal dynamics trajectory data \cite{kang_multiscale_2019}. Our approach expands the NSM scheme to a spatial setting where individual cells are represented as well-mixed reaction voxels or volumes, and the cell-cell communication via signaling-molecule diffusion is modeled as a stochastic jump process between cellular neighbors; see \cite{ramirez-sierra_ai-powered_2024} for additional details.

\subsubsection*{Simulation data processing and feature engineering}

Our simulator generates multiscale high-dimensional time series data. However, these raw synthetic high-dimensional observations are generally ineffective for model parameter inference due to the curse of dimensionality and the lack of experimental datasets with fine-grained resolution. Simulators do not capture all granular characteristics of true data-generating processes, and experimental studies are usually incapable of simultaneously measuring all critical system variables necessary for mechanistic representations \cite{verdier_simulation-based_2023}.

This problem results in low interpretability, high misspecification, and poor generalizability of the underlying models \cite{tolley_methods_2024}. An attractive solution to this challenge is feature engineering: transforming high-dimensional temporal-spatial stochastic trajectories into low-dimensional projections. While tools such as PCA and ANN embedding are increasingly popular for automatic extraction of features (or summary statistics) \cite{varolgunes_interpretable_2020, goncalves_training_2020, tolley_methods_2024, dingeldein_amortized_2024}, there is no full assurance that these tools will find all necessary and sufficient data conditions correctly mapping particular waveforms or time-series shapes to relevant model parameter values \cite{goncalves_training_2020, verdier_simulation-based_2023, tolley_methods_2024}.

As such, \emph{ad hoc} approaches can leverage unique domain knowledge to formulate (or handcraft) summary statistics tailored to the target system behavior. These handcrafted summary statistics must be used in conjunction with diagnostic analytics to track the proper progression of the posterior estimation scheme \cite{tolley_methods_2024}.

Here we engineer a suitable feature, referred to as the ``pattern score function'', which allows us to monitor posterior estimation performance via a unique scalar, referred to as the ``meta score''. To construct the pattern score function, we take advantage of three critical empirical characteristics of the studied biological system: (1) emerging ICM fates exhibit a highly reproducible ratio of $2:3$ for EPI and PRE lineages \cite{saiz_coordination_2020}; (2) expected EPI-PRE fate proportions are reached within a 48-hour time window, roughly 12 or 8 hours before the final time point preceding the implantation stage \cite{allegre_nanog_2022}; (3) inoperative cell-cell signaling via FGF4 stimulates almost all cells to adopt the EPI fate \cite{bessonnard_icm_2017}.

This third characteristic requires us to embrace the notion of multiple ``target rules'' for the ITWT system, referring to different model configurations. It reflects that naive cellular pluripotency (i.e., a potency akin to the EPI fate) is the default state of the ICM population when cell-cell signaling is lost through directed mutations. A biologically meaningful model of this system must simultaneously recapitulate both wild-type and mutant-like cases. This notion necessitates running simulations for both model configurations in parallel to find relevant parameter value sets attaining and sustaining both goal ratios. This procedure can be extended to three or more rules as shown next (see \nameref{subsubsection:multiple_model_configuration_combination}). Note that for the RTM system, only one target rule applies, as it already lacks FGF4-mediated cell-cell signaling by design.

\subsubsection*{Pattern score defines a suitable objective function}

We simulate, via parallel computing, 48-hour composite (event-driven) trajectories of all complementary model configurations for grids of 25 cells; see again Fig~\ref{fig1}[A]. Using sequential data transformations, we first reshape the simulation trajectories onto a regular time mesh with a sampling period of 0.25 hours, which is sufficient for capturing the relevant tissue-scale temporal dynamics of the considered process.

Each simulated cell has an associated multidimensional time series representing the dynamics of all biomolecular counts; these counting variables trace the gene regulatory network (GRN) elements and the signaling pathway components.

Although exact simulations require tracking all involved biochemical species, we focus only on two observables that fully characterize target system behavior (correct cell-fate specification) at the cellular scale: NANOG, a key protein marker of the EPI lineage, and GATA6, a key protein marker of the PRE lineage. These two protein counts are compared against predefined copy-number thresholds for classifying individual cells into one of three fates: EPI lineage (NANOG $\ge$ 388 and GATA6 $\le$ 329), PRE lineage (NANOG $\le$ 329 and GATA6 $\ge$ 842), and UND (undifferentiated) lineage otherwise.

At the tissue scale, full characterization of target system behavior (correct cell-fate proportioning) requires three observables at each simulated time point: EPI, PRE, and UND population counts. These three cell-counting variables are the most important components of our feature engineering pipeline. In other words, the assigned cell fates are converted into cell-fate population counts at the tissue scale to fully characterize the target system behavior.

Without loss of generality, for each fate at every simulated time point, the pattern score time series is the output of a vector-valued objective function, whose two inputs are total and target cell-fate counts per model configuration or rule.

Formally, we index the three possible cellular population fates as 0, 1, and 2 for EPI, PRE, and UND, respectively. Let $\boldsymbol{Z}_{t}=(Z_{t,0},Z_{t,1},Z_{t,2})$ be a discrete random vector that takes values in $\mathbb{N}^{3}$, representing the total cell-fate counts at a particular discrete time point $t \in \mathbb{N}$. Let $\boldsymbol{w}_{m}=(w_{m,0},w_{m,1},w_{m,2}) \in \mathbb{N}^3$ be a discrete vector representing the target cell-fate count for a particular model configuration; we index the possible rules by $m \in \{0,1,\ldots\}$.

Thus, for each model configuration, the (marginal) pattern score is a continuous random variable $S_{t,m}$ representing a nonlinear transformation mapping from an absolute-difference vector $|\boldsymbol{z}_{t}-\boldsymbol{w}_{m}|=(|z_{t,0}-w_{m,0}|,|z_{t,1}-w_{m,1}|,|z_{t,2}-w_{m,2}|)$ to a point $s_{t,m}$ in the closed interval $[0,1] \in \mathbb{R}$:

\begin{align} 
S_{t,m} & = \frac{\exp{\left(-{\|\boldsymbol{Z}_{t}-\boldsymbol{w}_{m}\|}_{1}/{\|\boldsymbol{w}_{m}\|}_{1}\right)}-\overset{\star}{w}}{1-\overset{\star}{w}}
\\[13pt]
\overset{\star}{w} & = \exp{\left(-2\max{\left(\boldsymbol{w}_{m}\right)}/{\|\boldsymbol{w}_{m}\|}_{1}\right)}
\label{eq:marginal_score}
\end{align}
Here, ${\|\cdot\|}_{1}$ indicates the $\ell^{1}$ vector norm, and ${\|\boldsymbol{w}_{m}\|}_{1}/{\|\boldsymbol{Z}_{t}\|}_{1} = 1$.

\subsubsection*{Combining multiple model configurations into a joint pattern score} \label{subsubsection:multiple_model_configuration_combination}

If the desired system behavior is characterized by only one model configuration (like for the RTM), then the ``marginal'' pattern score $S_{t,m}$ is the unique and final feature necessary for discriminating the target system behavior. In contrast, when there are two or more model configurations (like for the ITWT), the multiple resulting (univariate) marginal time series must be combined into a single ``joint'' pattern score.

For this purpose, we perform parameter space exploration by exploiting two companion $\ell^{1}$- and $\ell^{2}$-inspired penalty methods for combining multiple model configurations, applied through two separate and independent runs for comparison. These methods penalize the arithmetic mean of all the marginal pattern scores by subtracting a correction value that stresses distinct search priorities: the $\ell^{1}$ (or L1) penalty judges parameter set quality based on the worst or smallest value among all the model configuration scores, while the $\ell^{2}$ (or L2) penalty judges parameter set quality based on a value representing the central tendency among all the model configuration scores but biased towards the minimum.

For the case of just two model configurations, the following formula completely describes both penalizing methods:

\begin{equation} 
S_{t} =
\begin{cases}
\begin{aligned}
\left(\frac{S_{t,0}+S_{t,1}}{2}\right) - \left|\frac{S_{t,0}-S_{t,1}}{2}\right| \quad & \textrm{if} \quad \ell^{1}, \\[13pt]
\sqrt{\left(\frac{S_{t,0}+S_{t,1}}{2}\right)^{2} - \left(\frac{S_{t,0}-S_{t,1}}{2}\right)^{2}} \quad & \textrm{if} \quad \ell^{2}.
\end{aligned}
\end{cases}
\label{eq:joint_score}
\end{equation}
Here, $\left|\frac{S_{t,0}-S_{t,1}}{2}\right|$ and $\left(\frac{S_{t,0}-S_{t,1}}{2}\right)^{2}$ are the possible penalty terms.

However, this formula does not consistently generalize or scale to three or more rules. Fortunately, it is easy to prove that, under simple assumptions, the $\ell^{1}$-inspired penalty method is equivalent to taking the minimum among marginal pattern scores, and the $\ell^{2}$-inspired penalty method is equivalent to taking the geometric mean among marginal pattern scores, which can be extrapolated to three or more rules; i.e.:

\begin{equation} 
\begin{gathered}
\begin{aligned}
\min(S_{t,0},S_{t,1}) & = \left(\frac{S_{t,0}+S_{t,1}}{2}\right) - \left|\frac{S_{t,0}-S_{t,1}}{2}\right|
\\[13pt]
\sqrt{S_{t,0}S_{t,1}} & = \sqrt{\left(\frac{S_{t,0}+S_{t,1}}{2}\right)^{2} - \left(\frac{S_{t,0}-S_{t,1}}{2}\right)^{2}}
\end{aligned}
\end{gathered}
\label{eq:joint_score_equivalence}
\end{equation}
Please see \nameref{section:supporting_information} for these proofs.

\subsubsection*{Meta score}

We monitor the posterior estimation performance via the ``meta score''. This meta score condenses into a single scalar the success of each parameter vector, with respect to the target system behavior. As such, the meta score serves as a stopping criterion for the AI-MAPE workflow, and it also serves as an exploration-guiding criterion for the SA-SGM workflow.

The meta score, $\overline{S}$, is a statistic calculated on a collection or ensemble of pattern score time series via the following formula:

\begin{equation}
\overline{S} = \operatorname{mean}\left(\max{\left(\boldsymbol{\alpha}_{50}-\boldsymbol{\beta},\vec{0}\right)}\right)
\label{eq:meta_score}
\end{equation}
Where $\overline{s}$ is the realization of the random variable $\overline{S}$: $\overline{s} \in [0,1] \subset \mathbb{R}$. The symbol $\boldsymbol{\alpha}_{50}$ denotes the 50-percentile (or median) vector of a particular pattern score ensemble at every simulated time point. The symbol $\boldsymbol{\beta} = ((\boldsymbol{\alpha}_{95}-\boldsymbol{\alpha}_{5})/2)^{2}$ denotes an element-wise penalty vector, intending to favor high cell-fate proportioning accuracy and precision. The symbol $\vec{0}$ denotes the zero or null vector. Finally, we apply the maximum operator to every component of the two input vectors and calculate the arithmetic mean.

We remark that the calculation of the meta score must be applied only to the last 16~h window of any simulation trajectory ensemble for this particular work. This time-window constraint is critical because it informs the ANN that the target behavior of the underlying biological system must be reached and sustained (within a reasonable buffer) 12 to 8 hours before the last simulated time point, as described previously, without considering any other past simulation state.

\subsection*{Inference workflow: AI-MAPE} \label{subsection:inference_pipeline_AI-MAPE}

We now illustrate the key elements of our AI-MAPE workflow. Our goal is not only to find a model parameter set that reproduces a biologically relevant target behavior, but also to discover compensatory mechanisms between model parameters and understand their sensitivity or robustness against value perturbations.

To extract this rich information cheaply and effectively, we opted for the sequential neural posterior estimation (SNPE) algorithm variant C \cite{greenberg_automatic_2019}, combined with maximum-a-posteriori probability estimation (MAPE) \cite{lueckmann_benchmarking_2021}. To facilitate the integration of the SNPE algorithm and the MAPE calculation into one workflow, we employed a state-of-the-art SBI toolbox \cite{tejero-cantero_sbi_2020}. This toolbox also assisted with other analyses, such as quantification of parameter uncertainties and computation of parameter correlations.

The SNPE algorithm is a member of the deep-learning simulation-based inference (DL-SBI) class. DL-SBI algorithms take advantage of artificial neural networks (ANNs) to solve inverse problems \cite{kolmus_tuning_2024} and navigate model parameter spaces efficiently. These techniques are especially useful for likelihood-free (or intractable likelihood) settings, employing forward-model simulations to assist the training of ANNs for directly approximating the likelihood itself \cite{papamakarios_sequential_2019, boelts_flexible_2022}, the likelihood-ratio \cite{miller_contrastive_2023}, or the posterior \cite{greenberg_automatic_2019, deistler_truncated_2022}. Although each approach has unique benefits and disadvantages, the SNPE algorithm's ability to directly approximate the posterior parameter distribution is arguably its most advantageous feature when applying scarce and expensive model simulations \cite{kaiser_simulation-based_2023, tolley_methods_2024}.

To this end, the SNPE approach exploits a deep learning architecture (conditional neural network estimator) for approximately encoding a functional mapping between model parameters and model simulations (or, as in our case, handcrafted summary features) through normalizing flows \cite{tolley_methods_2024}, thereby obtaining a direct estimation of the (multidimensional) posterior distribution over the parameter space \cite{greenberg_automatic_2019, goncalves_training_2020}.

In principle, a single large simulation dataset can be harnessed for training an ANN within only one round, employing a procedure known as amortized inference via NPE, a non-sequential variant of SNPE \cite{greenberg_automatic_2019, goncalves_training_2020}. However, even with a high simulation data budget, single-round inference is generally not efficient enough for obtaining a useful approximation of the model parameter posterior distribution \cite{boelts_flexible_2022, kaiser_simulation-based_2023}.

This problem can be solved by iteratively refining a posterior estimate employing a multi-round strategy via SNPE, biasing parameter search by focusing on a single target observation congruent with the feature engineering pipeline \cite{kaiser_simulation-based_2023}. Yet, this non-amortized solution comes with two further caveats. First, any other target observations might not be properly captured by the trained ANN \cite{greenberg_automatic_2019, boelts_simulation-based_2023}. Second, this procedure creates proposal distributions by adjusting posterior distributions between rounds, but when the model parameter space is high dimensional (having over 10 components), this adjustment usually leads to ``posterior leakage'': drawn samples often fall outside the proposal range, making it prohibitively expensive to utilize SNPE for highly detailed models \cite{boelts_simulation-based_2023, dirmeier_simulation-based_2024}. Several approaches have been proposed to overcome this last caveat \cite{deistler_truncated_2022, xiong_efficient_2023}, but here we use a simple-yet-effective self-devised alternative, described next.

\subsubsection*{AI-MAPE multi-round strategy}

In this study, we perform 4 inference rounds for RTM, and 8 or 10 inference rounds for ITWT (owing to its larger number of unknown parameters). For each round, we generate 100 thousand simulation trajectories. Once the simulation stage is complete at each round, the ANN training data is generated all at once using every available trajectory. If the current round number is 2 or above, we augment the ANN training by merging the simulation data from the current and the (immediately) previous rounds.

Utilizing the same dataset, we train exactly 20 ANNs in parallel at each round. Despite using the same ANN architecture, the goal is to exploit the inherent stochasticity of the actual inference stage to select the best possible posterior distribution: the ANN training algorithm employs mini-batch stochastic gradient descent, and the pattern score trajectories include implicit randomness. A similar approach has been recently applied to some analogous problems \cite{wang_massive_2019}.

A key component is constructing a synthetic target observation to correctly train the ANN and obtain a useful approximation of the full posterior distribution. Our data processing pipeline creates a suitable latent ``feature'' space representation of the raw simulation data in terms of ``scores''. Therefore, it is straightforward to construct a time series with support spanning the relevant simulation period (32-48~h), where each entry represents the value corresponding to the ideal target behavior.

Each trained ANN provides unrestricted access to samples of an approximate posterior distribution, making it feasible to calculate its model parameter MAPE. For each MAPE, we generate exactly 50 additional simulations and construct its associated empirical pattern-score (time-series) distribution. To select the best learned posterior distribution at the current round, which is implicitly conditional on the target observation, we assess the quality of each MAPE by computing the meta score over its associated ensemble of 50 additional simulations. The best MAPE (or the best posterior) is the one with the highest meta score.

The posterior with the highest associated meta score is subsequently used as the prior for the next round. If the current-round meta score does not improve (i.e., it is smaller than or equal to the previous-round meta score), then this condition serves as a stopping criterion for the inference workflow, and no further rounds are conducted.

To avoid the problem of ``posterior leakage'' for our highly detailed mechanistic models, we feed the SNPE algorithm with a proposal that completely matches the original (first-round) prior. This remedy has proven effective and simple, though it may sacrifice some sampling efficiency.

\subsection*{Inference workflow: SA-SGM} \label{subsection:inference_pipeline_SA-SGM}

We now illustrate the key elements of our second workflow, the SA-SGM, which is based on the simulated annealing (SA) algorithm \cite{franzin_landscape-based_2023} and the sample geometric median (SGM) \cite{minsker_geometric_2023}.

The SA algorithm and its many variants belong to the metaheuristic optimization class \cite{franzin_revisiting_2019}. These algorithms are popular because they provide problem-agnostic guidelines that are easy to implement for complex (black-box) scenarios, while being effective for tackling computationally difficult optimization questions \cite{franzin_revisiting_2019, franzin_landscape-based_2023}. They can also be seen as stochastic local search algorithms \cite{franzin_landscape-based_2023}, where their success depends on balancing the explore-exploit dilemma, similar to reinforcement learning settings \cite{arumugam_bayesian_2023}. The searcher or agent must comprehensively explore the objective parameter space to discover informative regions, and efficaciously exploit available information to quickly reach its objective.

Moreover, the geometric median is a popular and robust estimator of location or central tendency for multivariate data, easily scaling from one- to high-dimensional distributions, serving as a generalization of the usual univariate median \cite{minsker_geometric_2015, minsker_robust_2016, minsker_geometric_2023}.

Adopting the concept of the geometric median and taking inspiration from the SA algorithm class (while controlling the exploration-exploitation trade-off), we implemented an inference workflow to challenge the AI-MAPE counterpart. We employed the common ``fast annealing schedule'' \cite{kochenderfer_algorithms_2019} and separated the inference problem into several rounds, just like in the AI-MAPE workflow: 4 inference rounds for RTM, and 8 or 10 inference rounds for ITWT. Given that it is always possible (and computationally inexpensive) to estimate the geometric median of an arbitrary empirical distribution \cite{minsker_geometric_2023}, the selection of the best discovered model parameter set relies on computing this estimate on adaptable sample sizes, a technique we simply call the ``sample geometric median'' (SGM).

In the following, we describe the SA-SGM multi-round inference strategy and our way of computing the SGM in detail.

\subsubsection*{SA-SGM multi-round strategy}

Per round, we initialize exactly 100 separate searchers in parallel, which independently sample the prior parameter distribution. Each searcher performs exactly 1000 exploration steps or model simulations, but the original prior is only involved during the first exploration step of the first round.

Once a current (single) simulation is complete, it undergoes the full data processing and feature engineering pipelines. We calculate its associated meta score, which guides the annealing scheduler towards the next exploration step. However, unlike classical SA variants, exploration does not stop within a round, even when the searcher discovers a parameter space point with the highest possible meta score (i.e. 1), in order to gather contextual information about local parameter neighborhoods.

At the end of a current-round simulation stage, when all the searchers have completed all their allowed exploration steps, we merge the collected information from all the agents into a (current-round) global dataset, which is yet split into two subparts. The first part covers the drawn parameters associated with the 90th or higher percentile of the empirical meta score distribution. The second part covers the remaining cut of the empirical meta score distribution.

There is no general routine for choosing an appropriate percentile threshold, so it acts as a hyperparameter that can be adjusted manually or heuristically between consecutive rounds. In our case, this percentile threshold also balances the explore-exploit trade-off: for a given searcher, the next-round model parameter set is uniformly randomly sampled from the current-round global (meta score) dataset, where the elements associated with the 90th or higher percentile are drawn with 90\% probability, and the remaining with 10\% probability. As such, the empirical meta score distribution acts as a next-round prior, and an exploration-guiding criterion for the SA-SGM workflow during successive rounds.

Hypothetically, a current-round inference stage would exploit the built (empirical) meta score distribution by picking the model parameter set associated with the highest current-round meta score. In case of a tie, any of the top-scoring sets is picked with equal probability. Just like in the AI-MAPE workflow, for this pick, we would generate exactly 50 additional simulations and construct its empirical pattern-score (time-series) distribution. If the current round number is 2 or above and the current-round meta score pick does not improve (i.e., it is smaller than or equal to the previous-round meta score pick), then this condition serves as a stopping criterion for the inference workflow, and there is no next round.

However, this hypothetical scenario is short-sighted and unreliable, as it only considers a single parameter sample at each round for making a decision about its objective knowledge of the explored space, without integrating any other cues. It does not ensure the most appropriate recapitulation of the target behavior for the underlying biological system. For these reasons, we discarded this hypothetical scenario and opted for an approach that comprehensively exploits the information collected during the exploration rounds, as detailed in the following section.

\subsubsection*{Sample geometric median}

Building on the concept of the SGM and exploiting the contextual information about local parameter neighborhoods gathered after completing all prescribed exploration rounds, we created an inference workflow that effectively leverages the parameter space data discovered and collected by all searchers.

To select the optimal point in parameter space (i.e., the best discovered model parameter set), we first construct a global dataset from all accessible parameter samples and their associated meta scores across all rounds. This all-round global dataset is systematically split into several (all-round) local subsets, aggregating parameter samples with similar meta scores into hierarchically increasing clusters; the number of clusters is an additional heuristic hyperparameter, see Fig~\ref{fig3}[A] or Fig~\ref{fig6}[C] for further clarification. Subsequently, for each all-round local subset, we find its SGM, which is, by definition, the parameter sample (vector) or point minimizing the sum of normalized distances from itself to the other vectors within the given subcollection.

For each found SGM, we generate exactly 100 additional model simulations and construct its empirical pattern-score (time-series) distribution, restricted to the 32-48~h interval. Following the production of this simulation score-data distribution, instead of computing its associated meta score, we characterize it using two closely linked statistics: the Time Ensemble Average Score (TEAS) and the Optimality Score (OS).

TEAS is a three-component vector statistic calculated by determining the 25th, 50th, and 75th percentiles of an SGM-associated simulation ensemble, and independently computing their time averages; i.e., $\emph{TEAS} = \left(\overline{\alpha}_{25},\overline{\alpha}_{50},\overline{\alpha}_{75}\right)$. OS is a scalar statistic calculated from TEAS via the formula $\emph{OS} = \overline{\alpha}_{50}(1-(\overline{\alpha}_{75}-\overline{\alpha}_{25}))$. It follows that the best or optimal model parameter set is the found SGM with the highest associated OS.

The functional forms of TEAS and OS are motivated by, and can be seen as modifications of, the concept of the meta score. At every simulation time point within the 32-48~h interval, TEAS estimates the accuracy of the target developmental behavior by evaluating the central tendency of pattern-score time-series ensembles using the median (i.e. 0.5 quantile), and estimates precision by calculating their interquartile range extreme points (i.e. 0.25 and 0.75 quantiles). OS combines these three statistics into a single scalar, rewarding high accuracy but penalizing low precision. Future revisions of our frameworks will consider fusing the meta score, TEAS, and OS, harmonizing all motivations into a unified measure.

\subsection*{Computational experiments} \label{subsection:computational_experiments}

\subsubsection*{Model parameter sensitivity matrix} \label{subsubsection:sensitivity_measure}

Sensitivity to value changes or perturbations for all model parameters is measured using two complementary approaches, depending on the nature of the estimated posterior distribution. Note that this analysis is computationally feasible by virtue of the ANN surrogate of the model simulator (via AI-MAPE). The first approach takes the raw posterior, which is consistent with the target observation, also known as the ``unconditional posterior''. The second approach takes the raw posterior conditioned on any parameter set, also known as the ``conditional posterior'', which can be either MAPE, SGM, or any other point estimate (value), as long as it is congruent with the prescribed prior distribution ranges. In addition, this second approach probes the model parameter space by altering one- or two-dimensional parameter values at once, while keeping other parameter interactions fixed at their point estimates.

In particular, with respect to the imposed prior ranges, we calculate the truncated posterior coverage for all parameter dimensions independently, relative to the histogram bin size (250 bins per dimension). The truncation filters the regions across one-dimensional posterior marginals (i.e. histograms) where the probability mass is greater than or equal to 0.25 (arbitrary threshold), thus retaining only the most significant parts and simplifying the overall analysis. An equivalent evaluation is conducted for the case of two-dimensional posterior marginals.

Furthermore, we employ a normalized logarithmic scale (base 10): a sensitivity coefficient of 0\% indicates that any parameter value within its prior range can recover the target observation, while a sensitivity coefficient of 100\% indicates that the parameter value must fall within a singular histogram bin to recover the target observation. Note that, for simplicity, the term ``parameter value'' refers to both singleton interactions (i.e. one-dimensional parameters) and duple or pair interactions (i.e. two-dimensional parameters); see Fig~\ref{fig4} and Fig~\ref{fig9}.


\section*{Results}

\subsection*{RTM: estimating primary core GRN interaction parameters for a toy system}

The main task explored in this work is the parameterization of a spatial-stochastic gene regulation model of stem cell differentiation in the inner cell mass (ICM) of the early mouse embryo. This parameterization should enable the generation of stochastic time series ensembles of the developmental dynamics that (1) reproduce experimentally determined features and (2) allow for quantification of the robustness and reproducibility of the cell-fate proportioning process. To this end, we constructed a highly detailed, biophysics-rooted spatial-stochastic model of the underlying gene regulatory processes and set out to determine the distributions of the involved biochemical parameters, as most of these have so far been inaccessible experimentally. To accomplish this, we chose Bayesian inference, which allows for identification of posterior distributions that make the model comply with prescribed empirical observations.

However, applying Bayesian inference in this context is challenging due to the inherent stochasticity of the considered system and the high-dimensional model parameter space \cite{goncalves_training_2020}. Moreover, validating such an inference workflow is difficult because ground-truth posterior distributions are typically inaccessible \cite{tolley_methods_2024}. This makes it hard to assess a priori which inference method will perform best for the given scenario.

To illustrate these challenges, we examine and compare two distinct inference workflows: AI-MAPE and SA-SGM, also referred to simply as AI and SA, respectively. AI-MAPE is guided by the sequential neural posterior estimation (SNPE) algorithm variant C \cite{greenberg_automatic_2019}, employing the maximum-a-posteriori estimation (MAPE) for optimal parameter set selection \cite{lueckmann_benchmarking_2021}. SA-SGM is guided by the simulated annealing (SA) algorithm \cite{franzin_landscape-based_2023}, employing a self-developed methodology centered on the sample geometric median (SGM) for optimal parameter set selection \cite{minsker_geometric_2015, minsker_robust_2016}.

We start by applying these workflows to a hypothetical \emph{toy model} featuring only four free parameters. This toy model, which we call the ``Reinferred Theoretical Mutant'' (RTM), is entirely derived from the full-fledged ``Inferred Theoretical Wild Type'' (ITWT) system (see Fig~\ref{fig2}). The RTM differs from the ITWT by its lack of cell-cell communication, such that the cell-lineage proportioning process relies exclusively on a probabilistic mechanism. Consequently, the RTM inherits all model parameters from the ITWT except for primary core gene regulatory network (GRN) interactions, which must be reestimated to establish the target cell-lineage proportioning.

We compare the two methods at equal computational power, performing the same number of inference rounds (four) and simulations per round (100 thousand). For each method, we also conducted two independent complete \emph{de novo} runs of the entire workflow, distinguishing them by indexes 1 and 2. Therefore, we compare four independent inference runs for the RTM, designated as AI\_1, AI\_2, SA\_1, and SA\_2 in the following.

\subsubsection*{Selecting optimal parameter sets for SA runs: sample geometric median (SGM)}

Both the AI and SA methods estimate full posterior distributions, but actual models can only incorporate one unique set of parameter values, requiring some way to discriminate the best among all the plausible parameter sets. In Fig~\ref{fig3}, we provide an overview of the optimal parameter value set selection for the SA method. This involves two relevant and closely linked statistics, which can be computed for both methods. The Time Ensemble Average Score (TEAS) (Fig~\ref{fig3}[A] top row) is a three-component vector calculated by restricting the simulation score data to the 32-48~h interval, determining the ensemble 25-50-75 percentiles, and computing their time averages separately: $\overline{\alpha}_{25}$, $\overline{\alpha}_{50}$, and $\overline{\alpha}_{75}$. The Optimality Score (OS) (Fig~\ref{fig3}[A] bottom row) is a scalar calculated from TEAS via the formula $\emph{OS} = \overline{\alpha}_{50}(1-(\overline{\alpha}_{75}-\overline{\alpha}_{25}))$.

For AI runs, a natural choice for determining the optimal parameter set is the maximum-a-posteriori probability estimate (MAPE) \cite{lueckmann_benchmarking_2021, glockler_variational_2022, boelts_simulation-based_2023}. The MAPE is a point estimator of the mode of the posterior distribution; i.e., it estimates the point in the high-dimensional parameter space where the true posterior probability distribution attains its maximum value. As the trained artificial neural network (ANN) serves as a surrogate of the posterior, the employed SBI toolbox \cite{tejero-cantero_sbi_2020} provides an effective algorithm for calculating the MAPE. In this situation, drawing posterior samples from the trained ANN is computationally cheap; see Methods section \nameref{subsection:inference_pipeline_AI-MAPE} for complete details.

For SA runs, the same way of selecting the optimal parameter set via the MAPE is applicable, in principle. However, calculating the MAPE of a high-dimensional distribution is generally a difficult combinatorial optimization problem, often requiring drawing simulated samples around well-posed initial guesses, which significantly increases the computational cost.

For that reason, we opted for an alternative way to compute the optimal parameter point estimates, relying on the concept of the geometric median \cite{minsker_geometric_2015}. The geometric median is a generalization of the usual univariate median, scaling from one- to high-dimensional distributions, as well as a popular and robust estimator of location or central tendency for multivariate data \cite{minsker_robust_2016, minsker_geometric_2023}. Accordingly, we refer to this technique as the sample geometric median (SGM) approach, summarized as follows: We start by viewing sets or samples of parameter values drawn from the approximate posterior distribution as elements, or points, of an abstract vector space. We then filter the whole collection of points, retaining only the vectors associated with \emph{meta-score} values larger than or equal to predefined thresholds, where the \emph{meta score} quantifies proximity to the idealized \emph{target} system behavior. Using different thresholds, we construct multiple vector subcollections. For each subcollection of parameter value vectors, we extract a sample or point that minimizes the sum of distances from itself to the other vectors in the subcollection. In other words, we find the geometric median of a given vector subcollection corresponding to a particular meta-score threshold value. See Fig~\ref{fig3}[A] and Methods section \nameref{subsection:inference_pipeline_SA-SGM} for complete details.

For each AI run, TEAS and OS are calculated once from a dataset with 100 additional simulations performed using the respective MAPE parameter value set from the inferred posterior distribution. For each SA run, every subcollection of parameter value vectors has an associated TEAS-OS duple calculated from a dataset with 100 additional simulations performed using the respective SGM parameter value set from a filtered inferred posterior distribution. The optimal parameter set is the estimated SGM with the highest OS; see Fig~\ref{fig3}[A] bottom row.

To validate the SGM approach, we compared the corresponding point estimation to the MAPE for the AI method, where both approaches can be applied. As summarized in S~Fig~\ref{S1_Fig}, contrasting both approaches reveals strong prediction agreement, showing that they are indeed analogous for the AI method.

To conclude, simply taking the naive approach of picking the best possible meta-score by estimating the SGM only from the subcollection of parameter value vectors with a meta-score equal to 1 does not guarantee the best recapitulation of the target behavior for the underlying system; see again Fig~\ref{fig3}[A] bottom row.

\begin{adjustwidth}{-1.75in}{0in} 
\includegraphics[width = 5.325in, height = 5.325in]{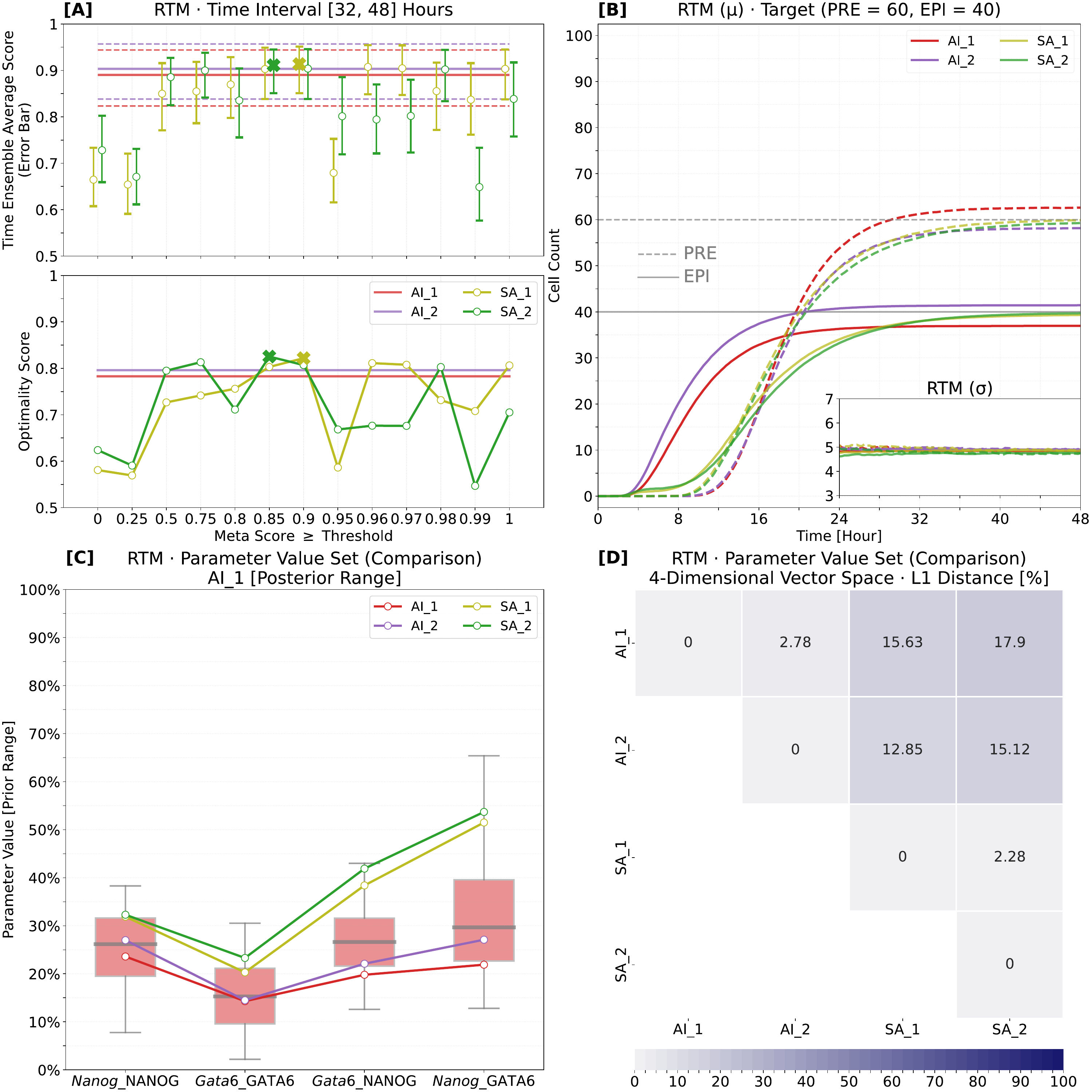} \centering 
\end{adjustwidth} 
\begin{figure}[hpt!]
\begin{adjustwidth}{-1.75in}{0in} 
\caption{{\bf Inferred optimal parameter sets for the RTM system.} {\bf [A]} Selection process for the best or optimal parameter sets of the two SA runs (SA\_1 and SA\_2). Statistics of the selection process for the optimal parameter sets of the two AI runs (AI\_1 and AI\_2) are shown for reference. The Time Ensemble Average Score ``TEAS'' (top row) is calculated by restricting the simulation score data to the 32-48~h interval, determining the ensemble 25-50-75 percentiles, and computing separately the time average of these three statistics: $\overline{\alpha}_{25}$, $\overline{\alpha}_{50}$, and $\overline{\alpha}_{75}$. Whiskers (SA) or dashed lines (AI) represent error bars; i.e., average interquartile ranges. Circles (SA) or solid lines (AI) represent average medians. The Optimality Score ``OS'' (bottom row) is calculated from TEAS via the formula $\emph{OS} = \overline{\alpha}_{50}(1-(\overline{\alpha}_{75}-\overline{\alpha}_{25}))$. Crosses highlight the best parameter sets for the SA runs. Notice that naively picking the best possible meta-score (filtering threshold equal to 1) does not directly translate to finding the best actual performance. {\bf [B]} Behavior at tissue scale for all four parameter sets (100-cell grid): the target of 60 PRE cells and 40 EPI cells should be reached and sustained within a time window between 32 and 48 hours. Note approximately equal standard deviations ``$\sigma$'' (inset plot), and similar intra- but distinct inter-method means ``$\mu$'' (main plot). {\bf [C]} Comparison of all four inferred optimal parameter sets: box-and-whisker diagrams for each AI\_1-related one-dimensional marginal posterior distribution are shown as baselines. Boxes cover interquartile ranges (from 25 to 75 percentiles) and show medians (50 percentiles). Whiskers cover ranges from 2.5 to 97.5 percentiles. Parameter values fall under normalized prior ranges. Notice intra-method similarity and inter-method differences. {\bf [D]} Distance matrix contrasting each pair of parameter sets: the parameter value sets are assumed to be elements of an abstract four-dimensional vector space. Notice once again short intra- but large inter-method distances. The L1 metric (normalized between 0\% and 100\%) was used to quantify distances between parameter vectors.}
\label{fig3}
\end{adjustwidth} 
\end{figure}

\subsubsection*{Behavior at tissue scale for the RTM system: all four parameter sets have high success rates}

To examine the tissue-scale behavior of the RTM system, we performed one thousand simulations per each optimal parameter set employing a 100-cell spatial grid; see Fig~\ref{fig3}[B]. For all four parameter sets, the target of 60 PRE cells and 40 EPI cells is reached and sustained within the required time window (32-48~h).

Note the approximately equal standard deviations ``$\sigma$'' (Fig~\ref{fig3}[B] inset plot), and the similar intra- but distinct inter-method means ``$\mu$'' (Fig~\ref{fig3}[B] main plot). Within this context, the standard deviation is related to the precision of the cell-fate proportioning process, and the mean is related to its accuracy. Although all four parameter sets have closely tied accuracy and precision values, the SA runs have a slight advantage in terms of accuracy.

\subsubsection*{Quantifying distances between optimal parameter sets: L1 metric}

To compare the concrete parameter values of the optimal sets, we analyze them per single dimension (Fig~\ref{fig3}[C]): note the significant intra-method similarity and considerable inter-method difference. We also inspect all their dimensions simultaneously (Fig~\ref{fig3}[D]): the L1 metric, normalized between 0\% and 100\% with respect to prior ranges, is used for quantifying distances between parameter vectors. Again, observe the short intra- but large inter-method distances.

Overall, despite conducting two independent runs per each method, we observe strong intra- but moderate inter-method agreement of the estimated parameter value sets. This observation highlights the innate biases from the two distinct inference workflows, using AI- and SA-based methods, as well as the two distinct optimal-parameter-set selection processes, using MAPE- and SGM-inspired approaches.

However, we remark that differences among the best parameter sets arise predominantly due to the number of available posterior samples; unlike AI-MAPE, SA-SGM distributions are constructed with limited exploration data (see \nameref{section:methods} section). Regardless of these differences, all the inferred posterior distributions display large overlapping regions in the RTM parameter space, as we describe in the following sections.

\subsection*{Analysis of inferred model parameter posterior distributions for the RTM system}

One of the goals of mechanistic modeling is to obtain a deep understanding of the studied system. As such, uncovering compensation mechanisms is a significant task, allowing for the assessment of the underlying system's flexibility with respect to parameter value perturbations. However, looking for compensation mechanisms is challenging in both experimental and theoretical settings.

By estimating full model parameter posterior distributions, one can, in principle, access a large population of models, assuming that a valid model can be obtained by simply sampling from such a posterior. The resultant population of models can be studied to determine how parameter interactions in the underlying system contribute to its target behavior \cite{marder_multiple_2011, jiang_identification_2022}.

\subsubsection*{Strong linear correlations among parameter pairs}

Since the AI method provides us with surrogate estimates of the posterior distributions akin to closed-form functions, it is possible to create hypotheses about parameter compensatory mechanisms by estimating linear correlation coefficients. These coefficients are extracted by conditioning the estimated posteriors, given a target observation, on the optimal parameter set, and varying parameters independently. Note that, when there is no access to an analytical posterior or some closed form of it, this exercise becomes an expensive computational task. Nevertheless, unlike SA, the AI method grants us access to unlimited samples from the posterior. Using the SBI toolbox \cite{tejero-cantero_sbi_2020}, we calculated a conditional (linear) correlation matrix for parameter singletons and duples for the AI\_1 run.

We observe that, while the unconditional posterior (given an ideal or target observation; see Fig~\ref{fig4}[A]) is rather broad with respect to the prior range, the conditional posterior (estimated posterior conditional on the best or optimal parameter set) is restricted to a small fraction of the parameter space (see Fig~\ref{fig4}[B]), indicating a high interdependence among parameter value choices.

Furthermore, the RTM system is simple (low dimensional) enough that the trained ANNs can predict similar matrices for all method-run pairs; see Fig~\ref{fig4}[C]. We observe that all these matrices produce correlation coefficients that strongly agree both in an intra- and inter-method fashion; see again Fig~\ref{fig4}. See also S~Fig~\ref{S2_Fig} for observing the weak linear correlation predicted by the unconditional AI\_1 posterior; similar plots are not shown for simplicity.

Using this information, we can hypothesize about potential compensation mechanisms in the RTM model. Intuitively, strong relationships among these four core GRN motif parameters are expected. However, the AI method prediction goes one step further: it not only captures these synergies, but the trained ANN (AI\_1) also provides reasonable estimates for the parameter space regions where the ideal system behavior is achieved despite parameter disturbances.

For example, we observe that reducing or increasing RTM self-activation parameters \\ (\textit{Nanog}\_NANOG and \textit{Gata6}\_GATA6) must occur in tandem, while weakening RTM self-activation interactions requires strengthening mutual-repression counterparts \\ (\textit{Gata6}\_NANOG and \textit{Nanog}\_GATA6) in a complementary manner.

This aspect highlights a problem closely related to the concepts of identification and degeneracy. Degeneracy refers to multiple parameter sets producing similar target system behavior, making it difficult to identify the best parameter set from the available information. Identification issues are especially relevant when datasets are limited or noisy \cite{jiang_identification_2022, bruckner_information_2024}.

This observation might contradict the notion of the best or optimal parameter set, but developmental-biological systems show high reproducibility and robustness to perturbations. Thus, the idea of a unique set of parameter values capable of driving the underlying system towards its goal states under drastically different conditions is not entirely in agreement with the empirical and theoretical understanding of developmental biology. Therefore, the concepts of degeneracy and identification play an important role in properly framing any computational understanding of living systems \cite{goncalves_training_2020, jiang_identification_2022, massonis_distilling_2023, boelts_simulation-based_2023, dingeldein_amortized_2024}.

\subsubsection*{High sensitivity for parameter singletons}

An important additional feature of the AI method is that it allows for straightforward estimation of model parameter sensitivities by employing the full inferred posterior distribution. More specifically, we use sensitivity analysis to determine how changes in parameter interactions influence target system behaviors \cite{marder_multiple_2011}. By analyzing the estimated AI\_1 posterior conditioned on each optimal parameter set independently, it is possible to obtain sensitivity coefficients for all runs (see Fig~\ref{fig4}[D] and \nameref{subsubsection:sensitivity_measure} for an explanation of our sensitivity measure).

In short, sensitivity signifies inverse tolerance to changes in parameter singleton or pair values while keeping all other parameters fixed at their MAPE or SGM values. In Fig~\ref{fig4}[D], the observed intra- and inter-method differences indicate that while parameter sensitivity is moderate for most parameter relations, parameter tuning is clearly dependent on the values of other parameters. In this sense, the AI\_1 posterior analysis predicts strong sensitivities (50-100\%) for both \textit{Nanog}\_NANOG and \textit{Gata6}\_GATA6 in the SA runs, as well as strong sensitivities for both \textit{Gata6}\_NANOG and \textit{Nanog}\_GATA6 in the AI runs.

To conclude, the choice of parameter values does not occur independently of others, and the strong correlations reflect this interdependence; see again Fig~\ref{fig4}[C]. This is not unexpected, as intuitively the free parameters of the GRN motif control important aspects of system dynamics, with confounding effects on each other. In this respect, the conditional sensitivity matrices provide a quantification of this intuition and allow us to predict which parameters require precise tuning to achieve and sustain the imposed target system behavior.

\begin{adjustwidth}{-1.75in}{0in} 
\includegraphics[width = 5.5in, height = 5.5in]{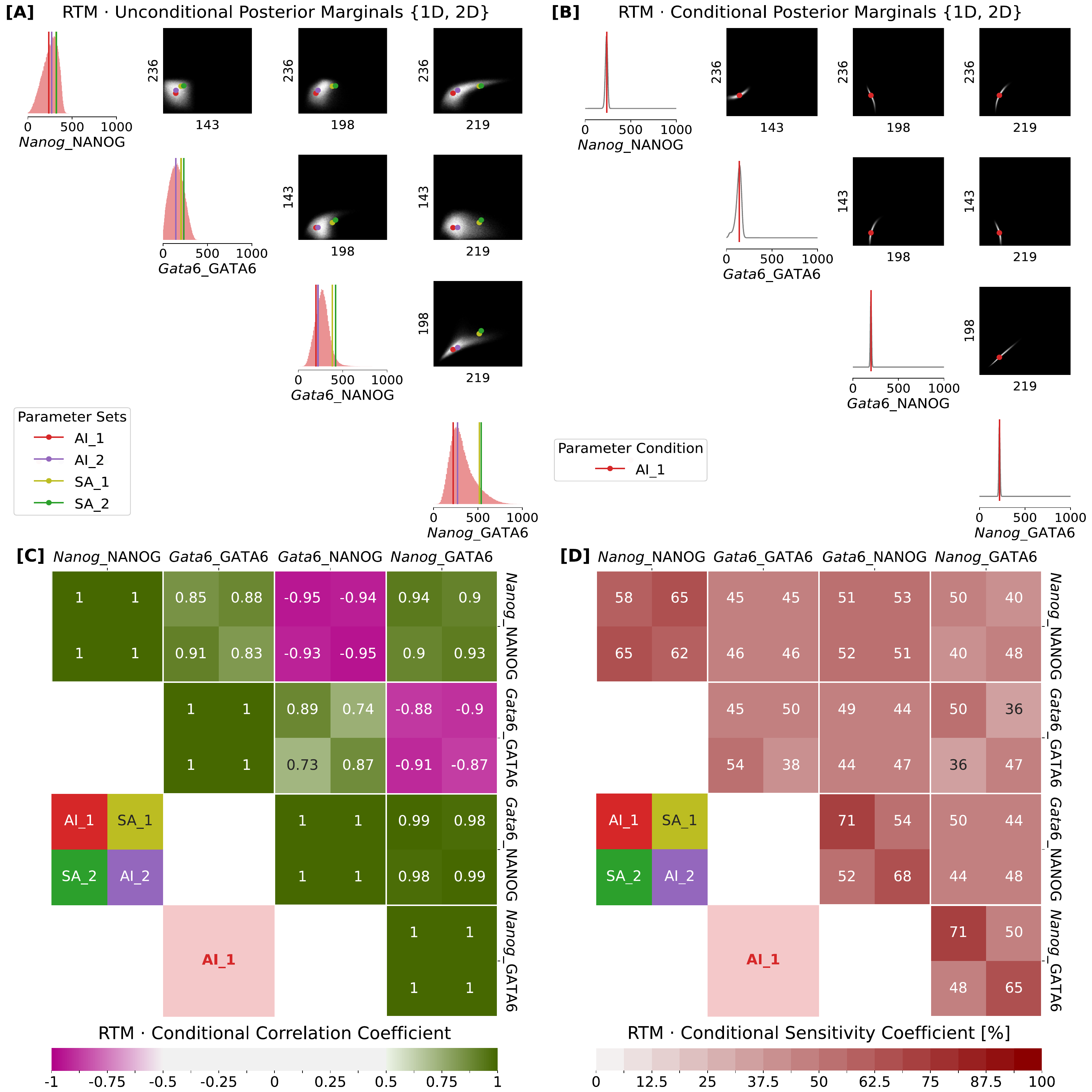} \centering 
\end{adjustwidth} 
\begin{figure}[hpt!]
\begin{adjustwidth}{-1.75in}{0in} 
\caption{{\bf RTM system. Analysis of inferred model parameter posterior distributions.} {\bf [A]} One- and two-dimensional marginal posterior distributions estimated from AI\_1 run for global reference. Locations of optimal parameter values with respect to prior ranges are shown for each method-run pair. {\bf [B]} One- and two-dimensional marginal posteriors conditioned on best parameter values (only referencing AI\_1 case); other cases are not shown for simplicity. {\bf [C]} Linear correlation coefficients extracted by conditioning AI\_1 posterior distribution (global reference) on every optimal parameter set separately. Note strong linear correlations among all parameter pairs (indicating potential compensatory mechanisms between each other), as well as strong intra- and inter-method agreement of estimated coefficients. Notice also that this calculation is possible because the trained ANN acts as a surrogate of the simulator and directly approximates the AI\_1 posterior; i.e., no additional simulations are required to extract Pearson's correlation coefficients. Despite significant differences among parameter value sets, the AI\_1-related ANN is capable of providing useful correlation information for all methods and runs. {\bf [D]} Sensitivity to value changes for all parameters; AI\_1 posterior conditional on every optimal parameter value set separately. Sensitivity equal to 0\% signifies that any value within its prior range can recover the ideal or target system behavior (while holding all the other parameters fixed). Sensitivity equal to 100\% signifies that the value must fall within a singular parameter bin (histograms were created with 250 bins per dimension). Note moderate sensitivities (25-50\%), and strong intra- but moderate inter-method agreement of estimated coefficients, for the majority of parameter relations. The AI\_1-related ANN predicts strong sensitivities (50-100\%) of both \textit{Nanog}\_NANOG and \textit{Gata6}\_GATA6 for the SA runs, as well as strong sensitivities of both \textit{Gata6}\_NANOG and \textit{Nanog}\_GATA6 for the AI runs, highlighting weak inter-method estimated-coefficient agreement among parameter singletons. No weak sensitivities (0-25\%) are recorded. See Results section for complementary details.}
\label{fig4}
\end{adjustwidth} 
\end{figure}

This analysis is in stark contrast to simply calculating parameter sensitivity on the unconditional distribution (see S~Fig~\ref{S2_Fig}), because in that case, insights into parameter sensitivities are limited to estimating the marginalized posterior coverage with respect to the prior ranges.

\subsection*{Analysis of estimated marginalized posterior distributions for the RTM system: differences between optimal parameter sets strongly correlate with differences between marginalized posteriors}

Understanding high-dimensional distributions is challenging. Many approaches for reducing dimensionality and facilitating analyses have been proposed for various contexts, such as PCA, UMAP, and t-SNE \cite{verdier_simulation-based_2023}. Although these approaches are relevant and useful in distinct cases, they also incorporate other nuances and complexities that can ultimately obscure the essence of parameter interactions.

Here, we opt for a simple approach to investigate similarities among intra- and inter-method predicted distributions by utilizing basic tools. Treating each estimated posterior as a probability vector in an abstract metric space, we use the Jensen-Shannon (JS) metric (with base 2) \cite{lin_divergence_1991} for estimating distances and quantify dissimilarities among all the predicted distributions.

In Fig~\ref{fig5}[A], visual comparisons between one- and two-dimensional posterior marginals are shown. Diagonal and lower plots also show distances with respect to the reference case ``AI\_1'' between one- and two-dimensional posterior marginals. Note that a distance of 0\% refers to minimally divergent histograms, and a distance of 100\% refers to maximally divergent histograms. Also note the large overlapping regions in RTM parameter space for the shown posterior projections.

In Fig~\ref{fig5}[B, C], mean or average distances between each method-run pair are shown. Overall, we observe the same synergy as for the comparison of parameter value sets (see also Fig~\ref{fig3}): significant intra-method similarity and considerable inter-method difference. Thus, significant differences between the posterior marginals estimated in distinct runs directly translate into significant differences between the optimal parameter sets determined from them.

\subsection*{ITWT: estimating all full core GRN and cell-cell signaling interaction parameters for a biologically meaningful system}

One important distinction between the RTM and ITWT systems, besides the full scope of the GRN motif and the functional cell-cell signaling, is the presence of multiple system configurations and target behavior dynamics in the ITWT system (see Fig~\ref{fig6}[A, B]). To capture this complexity, we defined a joint configuration score for each simulation and sampled parameter set, which is combined from the marginal configuration scores using an L1- or L2-inspired measure.

To quickly summarize the meaning of the L1- and L2-inspired measures: L1 judges the quality of the parameter set based on the worst value among the possible model configuration scores, while L2 does so based on a value that is at least as good as the worst score, but at most equal to the (arithmetic) mean value of all the model configuration scores. For these reasons, we compare four optimal parameter sets for the ITWT, designated as AI\_L1, AI\_L2, SA\_L1, and SA\_L2, indicating which combination of method and joint configuration score was used.

However, selecting the optimal parameter set per method-run pair can only be meaningfully done when focusing on the same optimality measure for all compared cases. Therefore, for simplicity, we opted for the L1 measure for this task (see Fig~\ref{fig6}[C]).

Initially, just like for the RTM, we compared both methods by imposing an ``equal computational cost'' restriction: every method run performs the same number of inference rounds (eight) and simulations per round (100 thousand). However, this restriction resulted in poor performance for the SA runs and (unexpectedly) AI\_L2; recall that the exact number of rounds needed to achieve the desired meta-score is unknown \emph{a priori}.

\begin{adjustwidth}{-1.75in}{0in} 
\includegraphics[width = 4.3in, height = 6.45in]{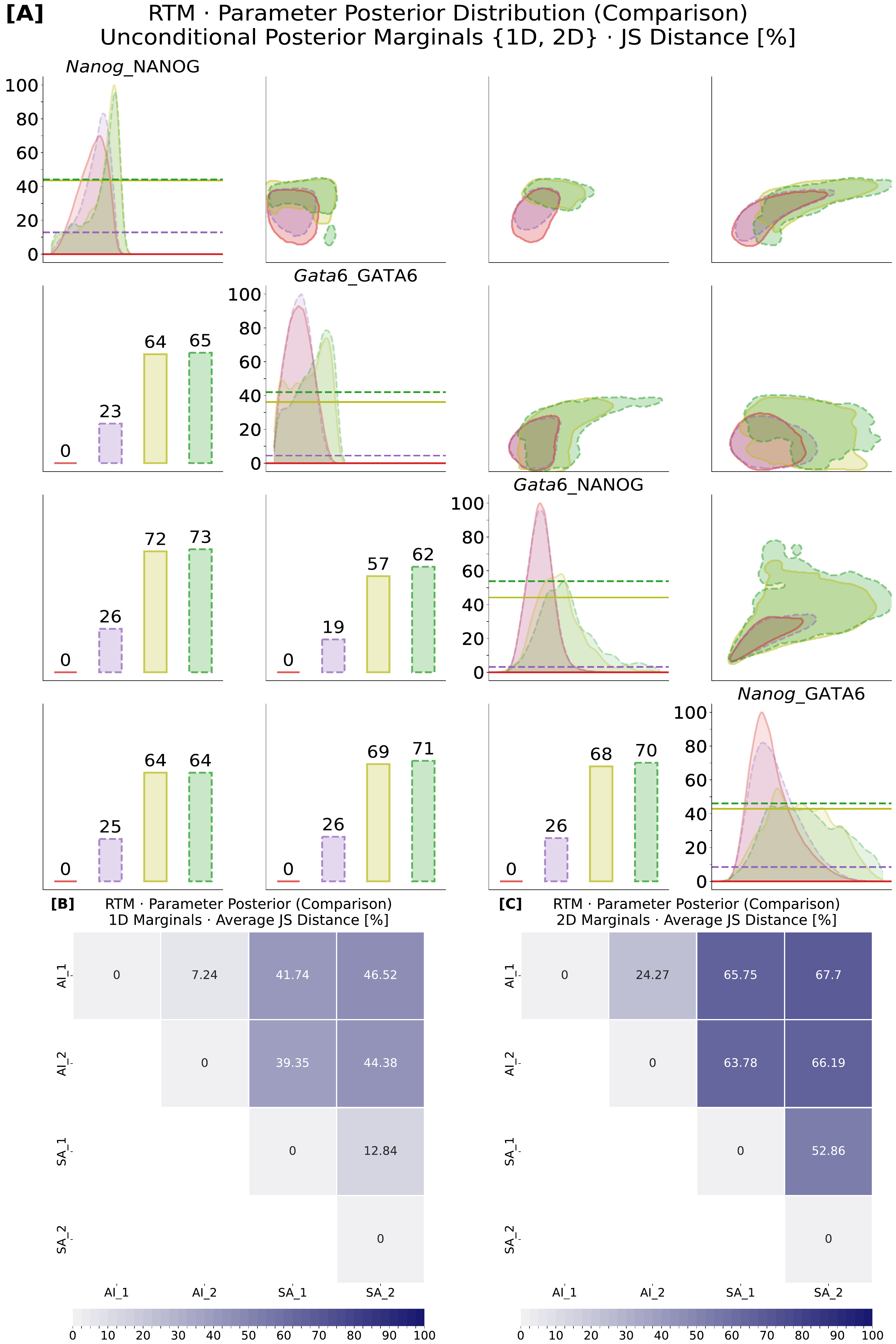} \centering 
\end{adjustwidth} 
\begin{figure}[hpt!]
\begin{adjustwidth}{-1.75in}{0in} 
\caption{{\bf RTM system. Analysis of estimated marginalized posterior distributions.} {\bf [A]} Comparisons between one- and two-dimensional (diagonal and upper elements) posterior marginals. Distances with respect to reference case ``AI\_1'' between one- and two-dimensional (diagonal and lower elements) posterior marginals are also shown. Distances are normalized between 0\% (minimally divergent histograms) and 100\% (maximally divergent histograms). Note that off-diagonal entries (lower and upper sectors) symmetrically correspond to one another. For upper-triangular entries, we filter histogram regions where probability masses are greater than or equal to 0.25 (arbitrary threshold), and create smooth projections via Gaussian kernel density estimates. {\bf [B, C]} Mean or average distances between each method-run pair. Notice the same synergy as for the comparison of parameter value sets (see Fig~\ref{fig3}): intra-method similarity and inter-method difference. Differences between the optimal parameter sets strongly correlate with differences between the marginalized posteriors. These differences come both from the two distinct posterior inference methods (AI versus SA), and the two distinct selection processes for the best parameter sets (MAPE versus SGM). Distances (shown as percentages) between raw probability vectors were quantified using the Jensen-Shannon metric (base 2). See also Methods and Results sections for additional details.}
\label{fig5}
\end{adjustwidth} 
\end{figure}

To enhance the statistical power of our workflow analysis, two additional rounds were performed for these three cases (AI\_L2, SA\_L1, and SA\_L2). Since AI\_L1 is already used in a companion study \cite{ramirez-sierra_ai-powered_2024}, it was natural to preserve it as our ``global reference'', setting a benchmark for the other cases. A performance-wise comparison between the equal-round-number cases is shown in S~Fig~\ref{S3_Fig}.

Interestingly, we observe stronger agreement among AI\_L1, AI\_L2, and SA\_L2 than between SA\_L1 and SA\_L2 (see Fig~\ref{fig6}[D]). This disagreement between SA runs highlights the effect of the curse of dimensionality, where limited sample sizes of SA runs make the distances between distributions meaningless as the available data points become sparse with increasing dimensionality.

\subsubsection*{Comparison of all four optimal parameter sets for the ITWT system}

In the top row of Fig~\ref{fig7}, we show the core GRN interaction parameters (cell-scale dynamics) obtained from the two independent AI inference and SA optimization runs. The bottom row also shows the corresponding spatial coupling (tissue-scale dynamics) and initial condition (mRNA-PROTEIN counts) parameters. We observe strong agreement among the parameter sets AI\_L1, AI\_L2, and SA\_L2 for predictions of primary core GRN interaction values (Fig~\ref{fig7} top-left panel), but weak agreement for predictions of secondary core GRN interaction values (Fig~\ref{fig7} top-right panel).

We also find moderate agreement among the predicted values for lifetimes, signaling, and initial condition parameters. These differences indicate the emergence of potential compensation mechanisms among parameters, highlighting distinct exploitable strategies to achieve the target system behavior. For example, \textit{Fgf4}\_NANOG (half-saturation threshold for transcriptional activation of the \textit{Fgf4} promoter by NANOG) and \textit{Fgf4}\_GATA6 (half-saturation threshold for transcriptional repression of the \textit{Fgf4} promoter by GATA6) are intuitively diametrically associated and should be adjusted in a correlated manner to tightly control FGF4 production. \textit{Fgf4}\_NANOG is significantly higher (weaker repression) than \textit{Fgf4}\_GATA6 for AI\_L1, AI\_L2, and SA\_L2, but this relationship reverses for SA\_L1. See also S~Fig~\ref{S4_Fig} for comparing all four ITWT optimal parameter sets at equal number of simulation-inference rounds.

\subsection*{Conditional parameter correlation matrix for the ITWT system uncovers parameter synergies influencing target dynamics}

Previous studies have demonstrated that the inferred full posterior distribution estimated via DL-SBI methods can potentially lead to novel scientific insights by revealing strong parameter interdependencies and predicting compensation mechanisms from the estimated posterior marginals \cite{goncalves_training_2020, boelts_simulation-based_2023, tolley_methods_2024}.

In this study, one such important and basic parameter interdependency is the relationship between ``mRNA'' and ``PROTEIN''; see Fig~\ref{fig8} (rightmost columns). Because the trained ANN acts as a surrogate for the simulator and directly approximates the AI\_L1 posterior, it is possible to extract (Pearson's) linear correlation coefficients by conditioning the AI\_L1 posterior distribution on its own MAPE without requiring additional, expensive simulations. Employing this feature, we calculated conditional correlations for all pairings of inferred parameters.

The two parameters ``mRNA'' and ``PROTEIN'' display a strong negative linear correlation, indicating that the AI\_L1-related ANN learned their complementary interaction: increasing initial \textit{Nanog}-\textit{Gata6} mRNA molecules should be compensated by proportionally decreasing initial NANOG-GATA6 protein copies, balancing the initial cellular resources across the tissue. These two parameters also strongly influence the initial \textit{Nanog}-\textit{Gata6}-\textit{Fgf4} (three main genes) expression dynamics, as reflected by their reciprocal relationships with \textit{Nanog}\_NANOG, \textit{Gata6}\_GATA6, \textit{Gata6}\_NANOG, \textit{Nanog}\_GATA6, \textit{Fgf4}\_NANOG, and \textit{Fgf4}\_GATA6; see again Fig~\ref{fig8}.

\begin{adjustwidth}{-1.75in}{0in} 
\includegraphics[width = 5.5in, height = 5.5in]{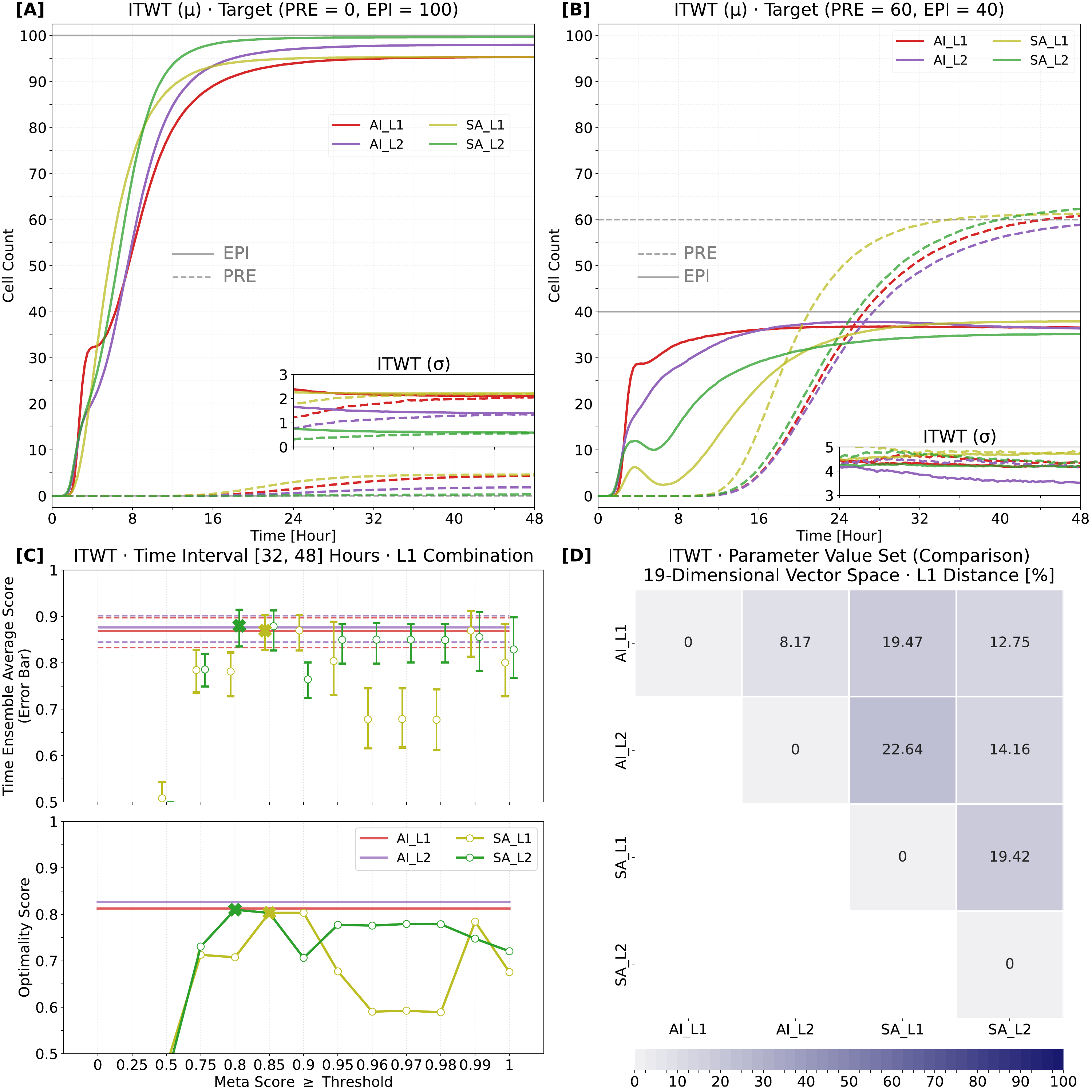} \centering 
\end{adjustwidth} 
\begin{figure}[hpt!]
\begin{adjustwidth}{-1.75in}{0in} 
\caption{{\bf Inferred optimal parameter sets for the ITWT system.} {\bf [A, B]} Behavior at tissue scale for all four parameter sets (100-cell grid). Correct cell-fate proportions should be reached and sustained within a time window between 32 and 48 hours for both system configurations separately. Temporal evolution of system dynamics reflects high target-fulfillment performance for all method-run pairs, strong similarities between AI runs, weak similarities between SA runs, and moderate agreement overall. {\bf [A]} The target of 0 PRE cells and 100 EPI cells applies to the configuration without cell-cell communication (nonfunctional signaling). {\bf [B]} The target of 60 PRE cells and 40 EPI cells applies to the configuration with cell-cell communication (functional signaling). {\bf [C]} Selection process for the best or optimal parameter sets of the two SA runs (SA\_L1 and SA\_L2). Statistics of the selection process for the optimal parameter sets of the two AI runs (AI\_L1 and AI\_L2) are shown for reference. Two techniques are used to produce the joint configuration-score data from the two marginal configuration-score time series (L1-norm- and L2-norm-inspired combinations), but calculating TEAS and OS for all method-run pairs requires a common technique: L1 combination is used for simplicity. See Fig~\ref{fig3} caption for definitions of Time Ensemble Average Score ``TEAS'' and Optimality Score ``OS''. Crosses highlight the best parameter sets for the SA runs. Notice that simply taking the naive approach of picking the best possible meta-score (filtering threshold equal to 1) does not directly translate to finding the best actual performance. See also Methods and Results sections. {\bf [D]} Distance matrix contrasting each pair of parameter sets: the parameter value sets are assumed to be elements of an abstract nineteen-dimensional vector space. Notice stronger agreement among AI\_L1, AI\_L2, and SA\_L2 than between SA\_L1 and SA\_L2. The L1 metric (normalized between 0\% and 100\%) was used to quantify distances between parameter vectors.}
\label{fig6}
\end{adjustwidth} 
\end{figure}

\begin{adjustwidth}{-1.75in}{0in} 
\includegraphics[width = 6.25in, height = 6.25in]{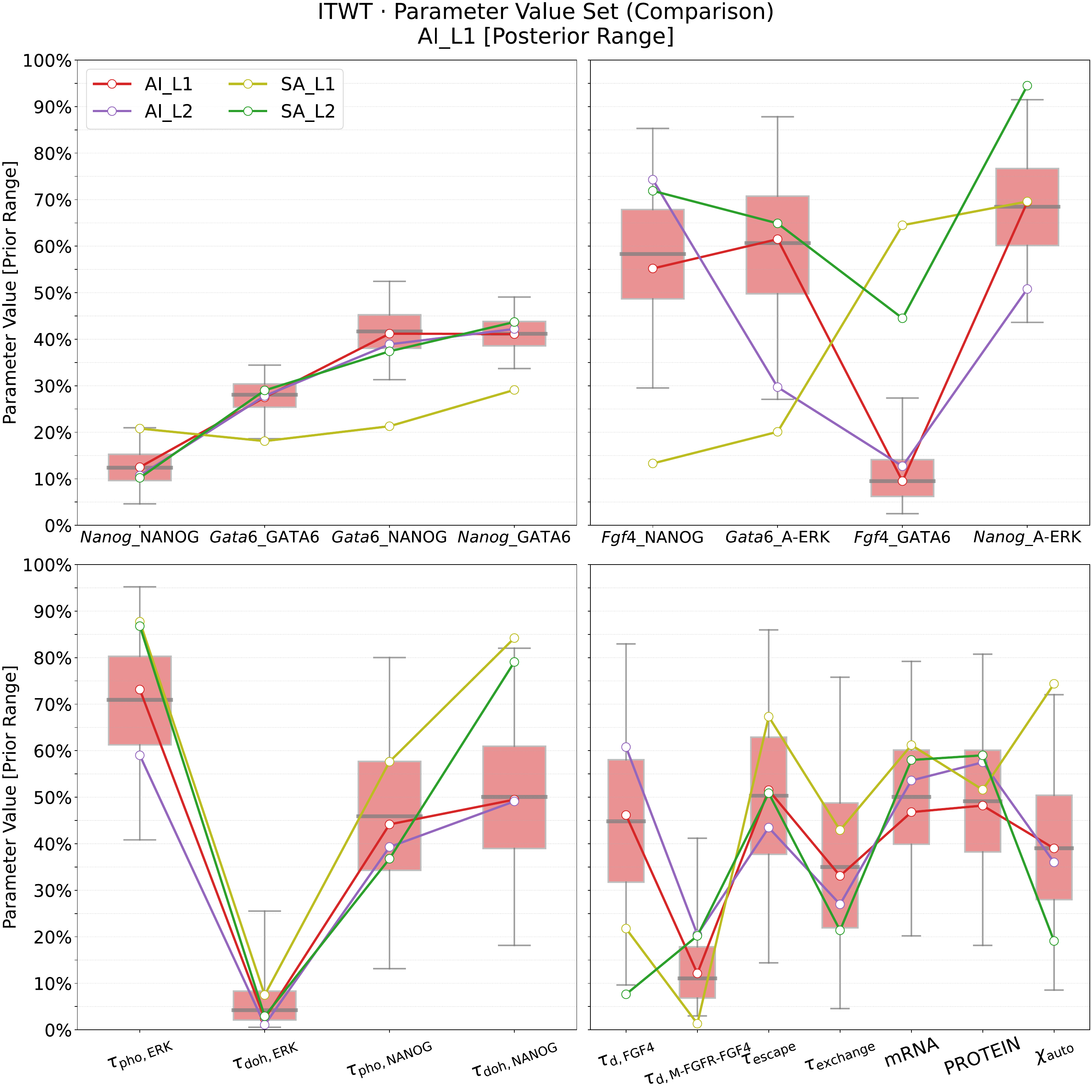} \centering 
\end{adjustwidth} 
\begin{figure}[hpt!]
\begin{adjustwidth}{-1.75in}{0in} 
\caption{{\bf ITWT system. Comparison of all four inferred/optimal parameter sets.} Box-and-whisker diagrams for each AI\_L1-related one-dimensional marginal posterior distribution are shown as baselines. See Fig~\ref{fig3}[C] caption for definitions of boxes and whiskers. Parameter values fall under normalized prior ranges. Top row: inferred core GRN interaction parameters (cell-scale dynamics). Bottom row: inferred spatial coupling (tissue-scale dynamics) and initial condition (mRNA-PROTEIN counts) parameters. Notice strong agreement among parameter sets AI\_L1, AI\_L2, and SA\_L2 for predictions of primary core GRN interaction values (top-left panel), but weak agreement for predictions of secondary core GRN interaction values (top-right panel). Note also moderate agreement among all parameter sets for predictions of other values. These differences indicate emergence of potential compensation mechanisms among parameters, highlighting distinct exploitable strategies to achieve the ideal or target system behavior. For example, \textit{Fgf4}\_NANOG (half-saturation threshold for transcriptional activation of \textit{Fgf4} promoter by NANOG) and \textit{Fgf4}\_GATA6 (half-saturation threshold for transcriptional repression of \textit{Fgf4} promoter by GATA6), intuitively, are diametrically associated and should be correlatively adjusted to tightly control FGF4 production. \textit{Fgf4}\_NANOG is significantly higher (weaker) than \textit{Fgf4}\_GATA6 for AI\_L1, AI\_L2, and SA\_L2, but this relation reverses for SA\_L1. See also Results section for additional details.}
\label{fig7}
\end{adjustwidth} 
\end{figure}

Several experimental and computational studies have concluded that replicable target behaviors can be achieved despite structural changes in variables or parameters associated with widely-studied biological systems \cite{boelts_simulation-based_2023, tolley_methods_2024}. This property can emerge for two reasons: (1) variability of some parameter subset minimally affects system behavior (low sensitivity); (2) variation of some parameters does noticeably influence system behavior, but compensatory mechanisms exist among them such that the target behavior is restored when they are changed in a concerted manner \cite{goncalves_training_2020}. This property can be further investigated by exploiting the information contained in the posterior distributions estimated via the AI method.

For instance, note the short and restricted ERK dephosphorylation time $\tau_{\emph{doh},\emph{ERK}}$ (all four predictions fall relatively close to each other), while ERK phosphorylation time $\tau_{\emph{pho},\emph{ERK}}$ is long and broad (all four predictions fall relatively far from each other). These two parameters not only dictate the speed of dynamics for signaling downstream of FGFR (FGF4 receptor) activation or inactivation, but also control the lifetime of NANOG proteins, thus affecting the system behavior in multiple aspects.

By studying the AI\_L1 posterior distribution, we observe not only the expected interdependency between ``mRNA'' and ``PROTEIN'' parameters dictating initial conditions, but also nuanced relationships between parameter pairs, highlighting the potential adaptability of the underlying biological system. The conditional correlation matrix shown in Fig~\ref{fig8} can be further explored and exploited to generate hypotheses potentially consistent (both qualitatively and quantitatively) with future experimental work or data \cite{goncalves_training_2020, tolley_methods_2024}; see the Discussion section for additional insights.

We highlight that this analysis is performed only for the AI\_L1 run but can easily be replicated for the AI\_L2 posterior. However, unlike for the RTM, the four inferred distributions for the ITWT are significantly distinct from each other (as will be shown later), limiting AI\_L1 posterior prediction power to only its own MAPE (optimal parameter set) or its distribution samples. We also highlight that these analyses are computationally unfeasible when employing non-DL-SBI approaches, as they might require a large number of auxiliary simulations to properly scan such high-dimensional parameter space \cite{goncalves_training_2020, boelts_simulation-based_2023, tolley_methods_2024}.

To summarize, the posterior distributions estimated via the AI method provide invaluable information for understanding nontrivial parameter dependencies. These inferred posterior marginals can help discover or hypothesize potential compensatory mechanisms among system parameters or variables without the need for additional simulations, which are typically required for a system-wide sensitivity analysis of this type.

\subsection*{Conditional parameter sensitivity matrix for the ITWT system reveals a subset of singular parameters strongly influencing target dynamics}

In addition to pairwise parameter correlations, the inferred posterior allows us to study sensitivity to value changes for all parameters by analyzing the AI\_L1 posterior conditional on its own MAPE, as summarized in Fig~\ref{fig9}. Overall, we answer a simple question: if all parameter values are fixed at their MAPE, how large are their predicted posterior coverages? \cite{tolley_methods_2024}. We observe a particular subset of parameters with strong (50-100\%) sensitivity to value changes, indicating high importance in obtaining ideal target behavior, namely: \textit{Gata6}\_GATA6, \textit{Gata6}\_NANOG, \textit{Nanog}\_GATA6, \textit{Fgf4}\_GATA6, $\tau_{\textrm{doh},\textrm{ERK}}$, and $\tau_{\textrm{d},\textrm{M-FGFR-FGF4}}$.

However, many parameter singletons and pairs display weak (0-25\%) or moderate (25-50\%) conditional sensitivity coefficients (see again Fig~\ref{fig9}). This observation does not signify that particular parameter subsets should not be finely tuned; these calculated conditional correlation coefficients are only associated with the AI\_L1 MAPE and strictly rely on the chosen prior ranges. Rich posterior distributions may be hiding sensitivity interdependencies among parameter values, as further illustrated in S~Fig~\ref{S5_Fig} for the conditional sensitivity matrix of AI\_L2. Parameter values have to be chosen concertedly to preserve underlying system behavior, as recorded by other studies \cite{boelts_flexible_2022, tolley_methods_2024}.

\begin{adjustwidth}{-1.75in}{0in} 
\includegraphics[width = 6.75in, height = 6.75in]{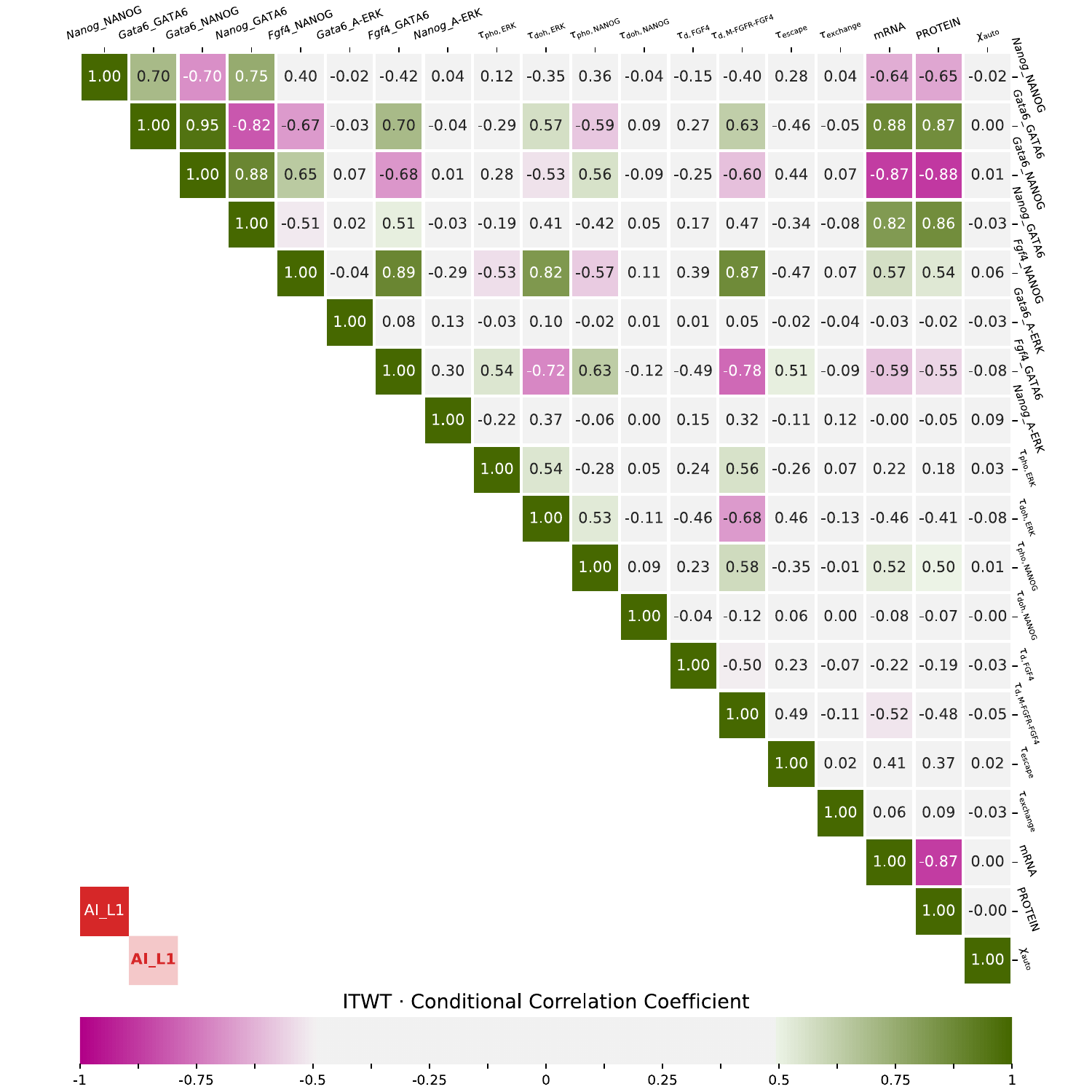} \centering 
\end{adjustwidth} 
\begin{figure}[hpt!]
\begin{adjustwidth}{-1.75in}{0in} 
\caption{{\bf ITWT system. Conditional model parameter correlation matrix.} Linear correlation coefficients extracted by conditioning AI\_L1 posterior distribution (global reference) on its own MAPE. Note that this calculation is possible because the trained ANN acts as a surrogate of the simulator and directly approximates the AI\_L1 posterior; i.e., no additional simulations are required to extract Pearson's correlation coefficients. Notice strong linear correlations between \{mRNA, PROTEIN\} and \{\textit{Nanog}\_NANOG, \textit{Gata6}\_GATA6, \textit{Gata6}\_NANOG, \textit{Nanog}\_GATA6\}. The two parameters ``mRNA'' and ``PROTEIN'' display a strong negative linear correlation between themselves, indicating that the AI\_L1-related ANN learnt their complementary interaction: increasing initial \textit{Nanog}-\textit{Gata6} mRNA molecules should be compensated by, proportionally, decreasing initial NANOG-GATA6 protein copies, balancing the initial cellular resources across the tissue. These two parameters strongly influence the initial \textit{Nanog}-\textit{Gata6}-\textit{Fgf4} (three main genes) expression dynamics, as reflected by their reciprocal relations with \textit{Nanog}\_NANOG, \textit{Gata6}\_GATA6, \textit{Gata6}\_NANOG, \textit{Nanog}\_GATA6, \textit{Fgf4}\_NANOG, and \textit{Fgf4}\_GATA6.}
\label{fig8}
\end{adjustwidth} 
\end{figure}

\begin{adjustwidth}{-1.75in}{0in} 
\includegraphics[width = 7.125in, height = 7.125in]{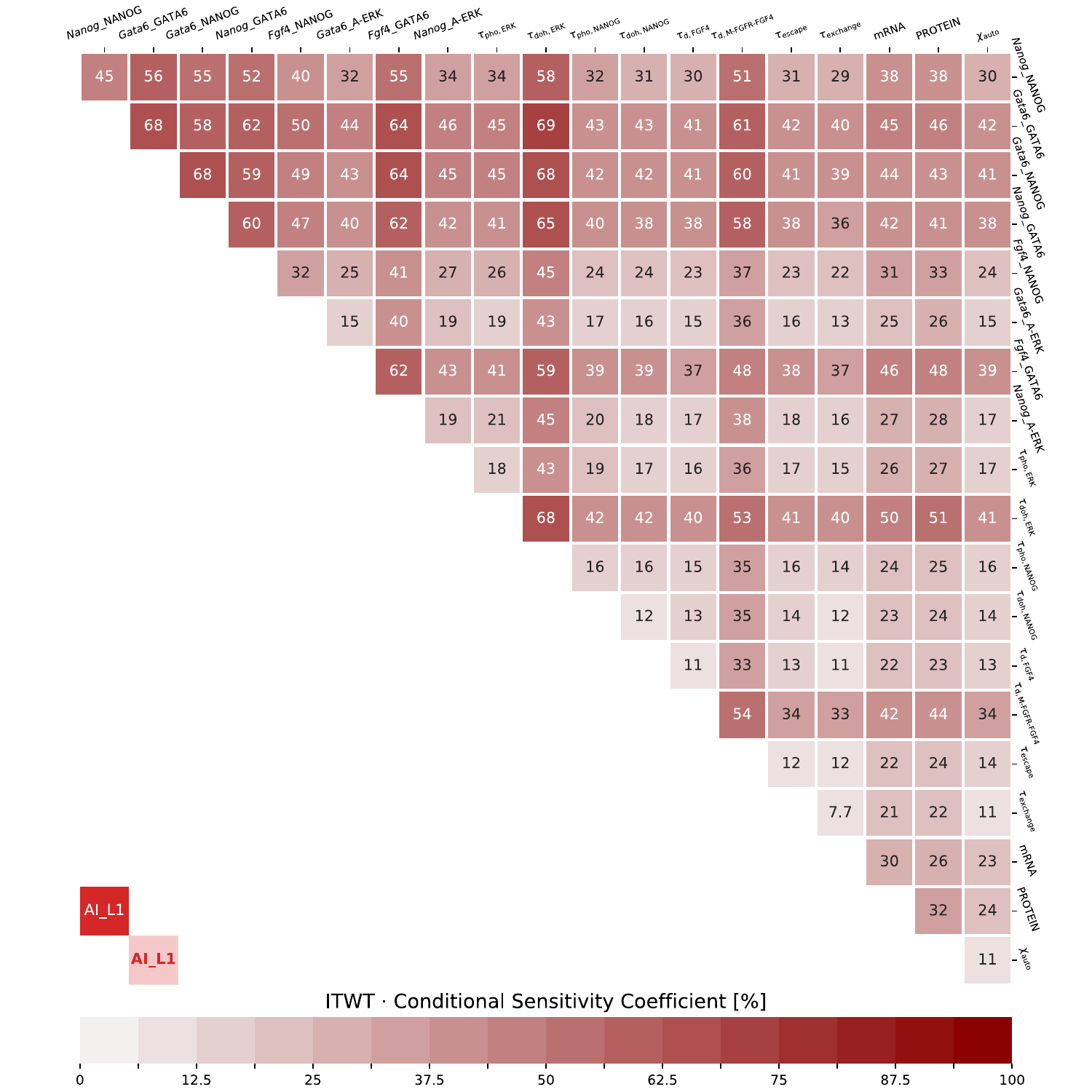} \centering 
\end{adjustwidth} 
\begin{figure}[hpt!]
\begin{adjustwidth}{-1.75in}{0in} 
\caption{{\bf ITWT system. Conditional model parameter sensitivity matrix.} Sensitivity to value changes for all parameters; AI\_L1 posterior conditional on its own MAPE. See Fig~\ref{fig4}[D] caption for definition of sensitivity. Note weak sensitivities (0-25\%) for (only diagonal entries) \textit{Gata6}\_A-ERK, \textit{Nanog}\_A-ERK, $\tau_{\textrm{pho},\textrm{ERK}}$, $\tau_{\textrm{pho},\textrm{NANOG}}$, $\tau_{\textrm{doh},\textrm{NANOG}}$, $\tau_{\textrm{d},\textrm{FGF4}}$, $\tau_{\textrm{escape}}$, $\tau_{\textrm{exchange}}$, and $\chi_{\textrm{auto}}$; moderate sensitivities (25-50\%) for \textit{Nanog}\_NANOG, \textit{Fgf4}\_NANOG, mRNA, and PROTEIN. Notice also strong sensitivities (50-100\%) for \textit{Gata6}\_GATA6, \textit{Gata6}\_NANOG, \textit{Nanog}\_GATA6, \textit{Fgf4}\_GATA6, $\tau_{\textrm{doh},\textrm{ERK}}$, and $\tau_{\textrm{d},\textrm{M-FGFR-FGF4}}$. Strong sensitivities reflect low tolerance to value fluctuations, given any other parameter is fixed at corresponding MAPE, for recapitulating the ideal target system behavior.}
\label{fig9}
\end{adjustwidth} 
\end{figure}

Consequently, choosing parameter values independently can be considerably problematic: one-dimensional posterior marginals may show reasonable parameter ranges, but the parameter space is high-dimensional, and potential interdependencies always require setting parameter values in concert. This fact alone highlights the benefit of a full posterior distribution, which can be easily sampled (in the AI method), being significantly more useful and informative than simply discovering parameters individually and not learning any relationships among them (as with the SA method) \cite{goncalves_training_2020, tolley_methods_2024}.

To conclude, for discovering fine-tuning requirements for parameter singletons or pairs, the conditional sensitivity matrix (Fig~\ref{fig9}) shows the calculated coefficients when all parameters are held constant at their MAPE (AI\_L1 posterior), while quantifying their posterior coverage with respect to the posterior ranges (one- and two-dimensional cases only).

\subsection*{Differences between optimal ITWT parameter sets strongly correlate with differences between marginalized posteriors}

Visual comparisons between one- and two-dimensional posterior marginals are shown in Fig~\ref{fig10}[A] for the primary core GRN motif parameters; see S~Fig~\ref{S6_Fig} for an extended plot of the core GRN motif parameters and S~Fig~\ref{S7_Fig} for a complete plot of the inferred parameter relations. Diagonal and lower panels also show Jensen-Shannon distances with respect to the reference case ``AI\_L1'' for one- and two-dimensional posterior marginals. Note that a distance of 0\% refers to minimally divergent histograms, while a distance of 100\% refers to maximally divergent histograms. Also note that, unlike in the RTM, overlapping regions in the ITWT parameter space are significantly more restricted when observing posterior projections (compare Fig~\ref{fig10} to Fig~\ref{fig5}).

In Fig~\ref{fig10}[B, C], mean or average distances between each method-run pair are shown. We observe the same synergy as in the comparison of parameter value sets; i.e., stronger agreement among AI\_L1, AI\_L2, and SA\_L2 than between SA\_L1 and SA\_L2 (compare Fig~\ref{fig6} and Fig~\ref{fig7}). In other words, the differences between the optimal parameter sets strongly correlate with the differences between the marginalized posteriors.

For some marginal projections, we observe parameter space islands (small for SA\_L2 run in Fig~\ref{fig10}[A] and large for both SA runs in S~Fig~\ref{S6_Fig}). Other computational studies have made similar observations \cite{marder_multiple_2011, goncalves_training_2020, boelts_flexible_2022, boelts_simulation-based_2023, tolley_methods_2024}, where seemingly disconnected parameter space regions result in identical system behavior found in such separated islands. However, they also note that, given sufficient posterior sampling, all relevant regions of parameter space appear to be fully connected, despite observing connecting paths with lower probability. This conclusion is directly linked to how well the ANN represents the objective parameter space and the particular structure of the parameter space for the underlying system.

\begin{adjustwidth}{-1.75in}{0in} 
\includegraphics[width = 4in, height = 6in]{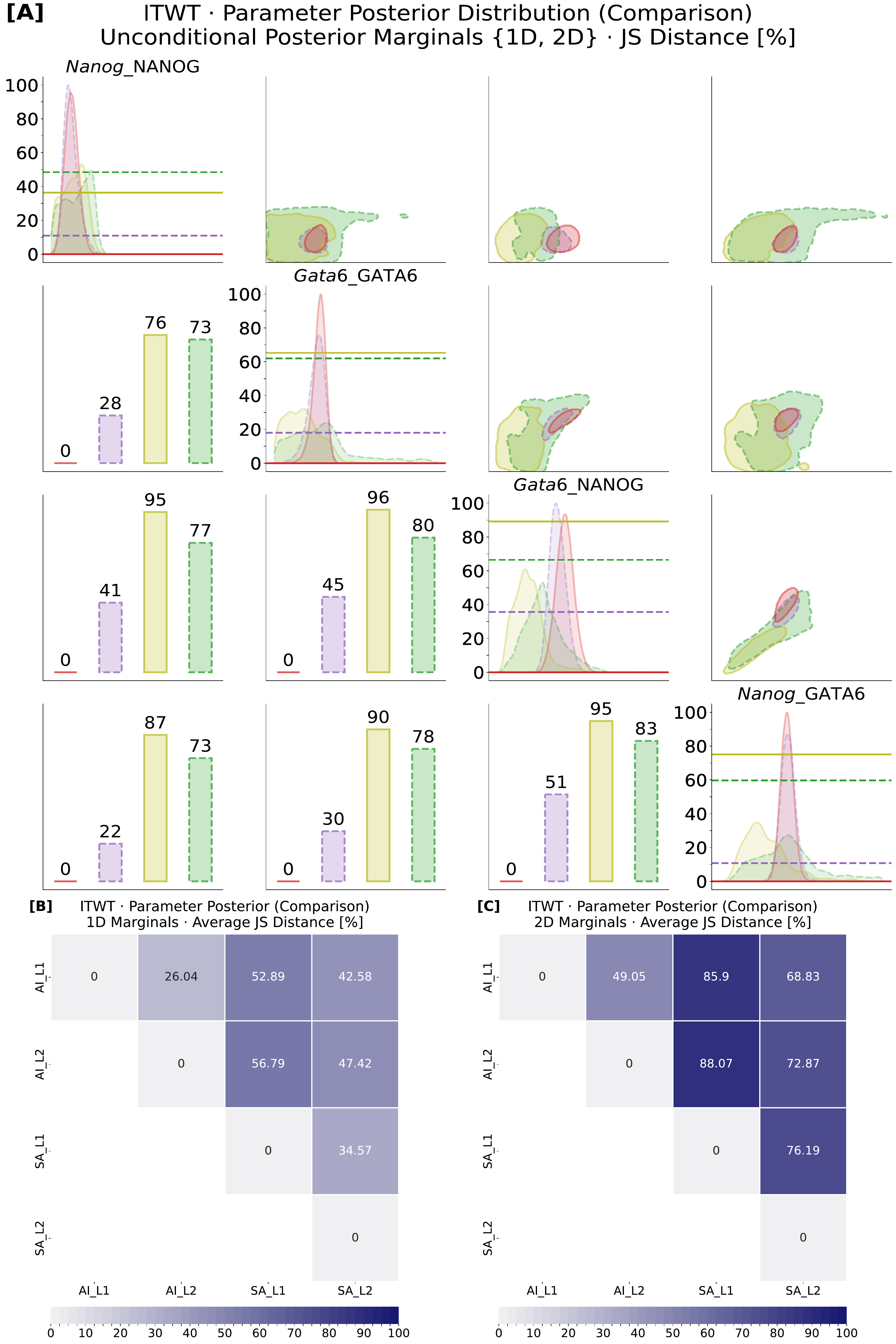} \centering 
\end{adjustwidth} 
\begin{figure}[hpt!]
\begin{adjustwidth}{-1.75in}{0in} 
\caption{{\bf ITWT system. Analysis of estimated marginalized posterior distributions.} {\bf [A]} Comparisons between one- and two-dimensional (diagonal and upper elements) posterior marginals. Distances with respect to reference case ``AI\_L1'' between one- and two-dimensional (diagonal and lower elements) posterior marginals are also shown. Distances are normalized between 0\% (minimally divergent histograms) and 100\% (maximally divergent histograms). Note that off-diagonal entries (lower and upper sectors) symmetrically correspond to one another. Only primary core GRN motif relations are shown; see S~Fig~\ref{S6_Fig} for extended plot of core GRN motif and S~Fig~\ref{S7_Fig} for complete plot of parameter relations. For upper-triangular entries, we filter histogram regions where probability masses are greater than or equal to 0.25 (arbitrary threshold), and create smooth projections via Gaussian kernel density estimates. {\bf [B, C]} Mean or average distances between each method-run pair. {\bf [C]} Notice the same synergy as for the comparison of parameter value sets; see Fig~\ref{fig6} and Fig~\ref{fig7}: stronger agreement among AI\_L1, AI\_L2, and SA\_L2 than between SA\_L1 and SA\_L2. The differences between the optimal parameter sets strongly correlate with the differences between the marginalized posteriors. These differences come from the two distinct posterior inference methods (AI versus SA), the two distinct selection processes for the best parameter sets (MAPE versus SGM), and the effect of the curse of dimensionality (limited sample sizes of SA runs). Distances (shown as percentages) between raw probability vectors were quantified using the Jensen-Shannon metric (base 2). See also Methods and Results sections for additional details.}
\label{fig10}
\end{adjustwidth} 
\end{figure}
\clearpage


\section*{Discussion}

In this work, we contrasted the performance of two distinct parameter inference workflows for spatial-stochastic gene regulation models. The AI-MAPE workflow, utilizing a state-of-the-art deep-learning technique fused with simulation-based inference (SBI), leverages a straightforward computation of the maximum a-posteriori estimate (MAPE) for optimal parameter set selection. The SA-SGM workflow, which elaborates on a classical simulated annealing (SA) optimization scheme and robust statistics, selects the optimal parameter set via the sample geometric median (SGM).

Both workflows satisfactorily inferred parameter sets for two complementary models of early embryonic cell-fate proportioning and maintenance. What is more, both inference workflows were capable of approximating full posterior distributions, creating informative pictures of model parameter spaces where high-probability regions are associated with the target behavior of the biological system under study.

Although both workflows displayed a similarly high performance in terms of accuracy and precision of the target (stem-cell proportioning) behavior, the classical SA-inspired method demonstrated an efficacy edge in terms of accuracy for the toy model case (RTM). In contrast, precision remained identical for both workflows (see Fig~\ref{fig3}[A, B]). For the biologically relevant model case (ITWT), the assessment is more complex: both workflows found solutions that appropriately recapitulate the target behavior, but their inferred parameter sets and posterior distributions were perceptibly more distinct, even between runs of the same workflow (see Fig~\ref{fig6}[A-C] and Fig~\ref{fig7}).

Remarkably, intra- and inter-method distances between parameter value sets and marginal posterior distributions (one- and two-dimensional projections) show strong correlation; see Fig~\ref{fig3}[D] versus Fig~\ref{fig5}[B, C], and Fig~\ref{fig6}[D] versus Fig~\ref{fig10}[B, C]. This indicates that they potentially serve as proxies for one another. This in turn suggests that observed differences among optimal parameter sets for all method-run pairs arise directly from the differences in estimated posterior distributions.

These results reflect the high complexity and stochasticity of the underlying biological system, as well as the curse of dimensionality. This problem is exacerbated in the classical SA-inspired method, as the high-dimensional parameter space cannot be properly characterized by the relatively restricted number of posterior samples associated with the ideal target behavior. The SA-SGM workflow collected only around 50 thousand posterior samples (out of 1 million total per run) consistent with high optimality scores. In contrast, the AI-powered method alleviates the curse of dimensionality by providing a virtually unlimited number of posterior samples consistent with high optimality scores, thanks to the deep generative artificial neural network (ANN) serving as a surrogate for the model simulator.

More importantly, while both workflows appear to perform sufficiently well in estimating the posterior, the AI workflow provides significantly more information than its classical counterpart, as already conceived by several other studies \cite{kaiser_simulation-based_2023, tolley_methods_2024, dingeldein_simulation-based_2023, stillman_generative_2023, dingeldein_amortized_2024}. By exploiting nuanced domain knowledge and carefully creating latent feature spaces, the SNPE algorithm directly and iteratively approximates full posterior distributions consistent with the target behavior. This approach yields waveform-free inferential solutions that fulfill imposed empirical constraints or observations via surrogate ANNs. This capability provides unique advantages for uncovering compensation mechanisms (e.g. Fig~\ref{fig8}) and quantifying estimation uncertainty (e.g. Fig~\ref{fig9}), while incurring minimal computational costs once ANN training is complete.

This is in stark contrast to traditional approaches such as gradient-based, genetic, evolutionary, or conventional ABC algorithms, which necessitate many more model simulation iterations, often employing brute-force parameter grid-search, to compute any of these instructive metrics \cite{goncalves_training_2020}. The AI workflow (powered by SNPE) emerges as a strong alternative, offering higher simulation efficiency compared to classical optimization methods, and enabling much richer insights thanks to its generative capabilities for posterior sampling tasks.

Moreover, by accessing AI-MAPE-estimated posterior distributions (especially AI\_1), we easily produced predictions of correlations and sensitivities for SA-SGM-related parameter sets, augmenting our model interpretation for the RTM (see Fig~\ref{fig4}). This exercise was not possible for the ITWT given its complexity and the parameter prediction discrepancies among all four runs: high-dimensional spaces expand distances between vectors, even when several components agree and are close to each other (see Fig~\ref{fig7}).

The hypothesized compensation mechanisms in this work could potentially guide future experimental studies. For example, the conditional correlation coefficient matrix for ITWT (Fig~\ref{fig8}) predicts that linearly modulating the effective lifetime of extracellular FGFR-FGF4 monomers (i.e., monomer receptor-ligand complexes) can be complementarily compensated by linear modulation of \textit{Fgf4} activation by NANOG and \textit{Fgf4} repression by GATA6. In other words, a shorter lifespan of the intercellular signaling molecule requires strengthening its gene activation and/or weakening its gene repression to preserve the target behavior. This observation is remarkable because of the nonlinear nature of the interactions between these mRNA and protein regulation components.

It is also reassuring that intuitive relationships, such as the association between the initial condition parameters (mRNA and PROTEIN), as well as the interactions among the primary core GRN motif, are already captured by the trained ANNs (see Fig~\ref{fig8}). This assessment is feasible thanks to a whole population of models retrieved via the ANN surrogate of the posterior, allowing us to explore the underlying structure of the parameter space. This capability enhances our understanding of compensatory mechanisms crucial for correct embryonic development and cellular function in real biological systems \cite{marder_multiple_2011}. These insights highlight that systematically and structurally different model parameter relationships (inputs) can give rise to identical target system behaviors (outputs).

Simultaneously, it is clear that although there may be degenerate solutions, they are unlikely to respond identically to parameter value perturbations. Many key environmental factors are usually inaccessible during the model selection process, complicating the discovery of the adaptability properties of the underlying system. Therefore, it is essential to filter degenerate solutions that not only are consistent with a target behavior but also respond robustly to external or internal perturbations, as instructed by our simulation-inference workflows. This filtering aims to recapitulate the strong resilience of true biological systems \cite{marder_multiple_2011, jiang_identification_2022, bruckner_information_2024}.

Similarly, we detected parameter space islands (small for the Fig~\ref{fig10}[A] SA\_L2 run and large for the S~Fig~\ref{S6_Fig} SA runs). Other computational studies have made similar observations \cite{marder_multiple_2011, goncalves_training_2020, boelts_flexible_2022, boelts_simulation-based_2023, tolley_methods_2024}, where point sets from seemingly disconnected parameter space regions reproduce the same system behavior. In our case, we do not believe that the presence of separate islands is simply an artifact of the regularization conditions imposed by the interpolation characteristics of the underpinning ANN; rather, it is an artifact of the low sample number obtained for the estimated SA posterior distributions. However, supporting this idea further will require performing considerably more simulations.

In essence, the effect of the curse of dimensionality is not only related to the limited sample sizes of SA runs (which itself constrains the analysis of posterior information), but also associated with a lack of inter-method information exchange and exploitability of SA posteriors. The last point is relevant because the small overlapping regions among all runs hinder the capabilities of AI-method-inferred posteriors in augmenting the value of the predictions done by SA-method-inferred posteriors; in contrast, this was not the case for the toy (RTM) system.

Nevertheless, while the potential compensation mechanisms uncovered by the AI method are insightful, their scope still hits certain limitations. Correlation and sensitivity analyses performed in this study are only possible for one- and two-dimensional marginalized posteriors. This poses a challenge, as the underlying distributions are high dimensional. The employed SBI toolbox \cite{tejero-cantero_sbi_2020} does provide some ways to produce higher-dimensional marginals than the ones explored here, but this requires using additional MCMC sampling steps, canceling the benefits of having access to ANN surrogates of posteriors. It might be plausible to enhance these analyses by exploiting embedding tools such as t-SNE and U-MAP \cite{verdier_simulation-based_2023}, and this could be a potential research direction.

Overall, it is interesting to think about combining both methods, AI-MAPE and SA-SGM. We propose the following straightforward approach: retain the searchers from SA-SGM but use the ANN posterior surrogate from AI-MAPE. This combination allows the contextual information gathered by exploring local parameter neighborhoods to be exploited and generalized for learning the global parameter space structure, while focusing on parameter regions with a high probability of recapitulating target observations.

More concretely, we propose combining intra-round biased sampling with inter-round exploration-information exchange enriched via diverse deep neural density estimation techniques such as (S)NLE \cite{papamakarios_sequential_2019}, (S)NRE \cite{miller_contrastive_2023}, (S)NPE \cite{greenberg_automatic_2019}, or (S)NVI \cite{glockler_variational_2022}. This approach potentially enhances the sampling efficiency of inference workflows by naturally and quickly guiding searchers towards optimality zones while extrapolating their parameter space configurations.

To select a high-quality next-round proposal distribution, instead of using just one ``optimal parameter set'', we propose utilizing both the MAPE and a weighted version of the SGM (WSGM). In the WSGM scheme, the coefficients would be set using all meta scores calculated from previous-round exploration steps. These two estimates could then further inform the underlying model for generating \emph{de novo} simulation ensembles to quantify their ``optimality'' via a revised version of the Time Ensemble Average Score (TEAS) and the Optimality Score (OS). Note that TEAS represents the accuracy of the target behavior using the median of pattern-score time-series ensembles and conveys the associated precision using the interquartile range. OS transforms TEAS into a single scalar, aiming to select ensembles with high accuracy and high precision. This combined approach leverages the strengths of both AI-MAPE and SA-SGM, potentially leading to more efficient and insightful parameter inference for complex biological models.

In conclusion, our comparison demonstrates that highly detailed, biophysics-rooted models of spatial-stochastic gene regulatory systems, which are prevalent across developmental biology, can be inferred by exploiting SBI approaches and high-performance computing, despite the lack of sufficiently granular experimental or empirical quantitative data. Additionally, it highlights an essential advantage of novel AI-driven SBI frameworks: access to ANN surrogates of posterior parameter distributions. By virtue of their generative capabilities, these ANN surrogate models significantly facilitate the discovery of crucial mechanistic insights for complex biological systems. This property underscores the vast potential applicability of SNPE-like algorithms to many other similar problems.

\clearpage 




\section*{Supporting information} \label{section:supporting_information}

\subsection*{Mouse ICM-lineage differentiation modeling} \label{subsection:stem_cell_lineage_differentiation}

We integrate the main gene regulatory processes controlling cell-lineage specification and proportioning for the mouse ICM-derived progenies (epiblast ``EPI'' and primitive endoderm ``PRE'') under spatially inhomogeneous conditions, while bypassing any description of mechanical interactions among cells (such as cellular division, proliferation, or cell motility) that act as extrinsic noise factors.

This simplification is necessary not only for computational feasibility (spatial-stochastic gene regulatory models are considerably more expensive to simulate than their ODE-based counterparts), but also for properly quantifying the synergistic effects of biochemical signaling at the cell and tissue scales without the impact of significant extrinsic variability. Unlike ODE-inspired models, explicitly spatial-stochastic models of gene regulation recapitulate non-instantaneous cell-cell signaling, while naturally incorporating the intrinsic noise arising from biochemical reactions and the low molecular abundances of key cellular resources.

Current experimental understanding suggests that this cell specification establishes a $2:3$ ratio between the EPI and PRE fates in a spatially uniform fashion across the ICM tissue \cite{saiz_coordination_2020, allegre_nanog_2022}, such that the emerging pattern does not store any positional information. Moreover, the related cell-fate proportioning mechanism has been observed to be robust and reproducible under widely distinct experimental settings (in vivo, in vitro, and organoid cases \cite{saiz_growth-factor-mediated_2020, raina_cell-cell_2021, fischer_salt-and-pepper_2023}), even when other potentially relevant embryo components and geometrical constraints are missing, such as the trophectoderm and blastocoel \cite{ryan_lumen_2019, shahbazi_mechanisms_2020}. This observation highlights the significance of biochemical signaling for correct developmental progress in this system, and the need to adequately quantify its role regardless of varying (extrinsic) environmental conditions.

Thus, our modeling approach focuses on the specification-proportioning process of EPI and PRE lineages from the ICM progenitor population. We constructed a biophysics-inspired, stochastic-mechanistic description of its (cell scale) gene regulation network and its (tissue scale) diffusion-based communication. The main drivers of the ICM differentiation process are the self-activation of \textit{Nanog} and \textit{Gata6} genes (primary markers of EPI and PRE lineages, respectively) together with their mutual repression \cite{saiz_coordination_2020}. Another fundamental driver of this process is an FGF4-facilitated feedback loop, which enables cell-cell communication to control the associated cellular fate proportioning \cite{plusa_common_2020}.

Our companion study \cite{ramirez-sierra_ai-powered_2024} delves deeply into the biological connotations of the parameter interactions uncovered via an AI-enhanced simulation-based inference workflow, which is identical to the AI-MAPE method treated here. These parameter interactions provide mechanistic insights into the realization and maintenance of the target behavior for the underlying developmental system.

\newpage 

\subsection*{Combination proofs} \label{subsection:combination_proofs}

Let $\rho = (\rho_{1}, \rho_{2}, \ldots, \rho_{m-1}, \rho_{m})$ be a finite random vector whose components are independent but not necessarily identically distributed, and such that $\forall j \leq m$ ($m \in \mathbb{N}$) each random variable $\rho_{j}$ maps to the closed real interval between 0 and 1; i.e., $\rho_{j}: \Omega \rightarrow [0, 1]$, where $\Omega$ is the event space.

Assume that there are two alternative ways to pairwise combine such random vector components: Eq~(\ref{eq:elle_one_combination}) and Eq~(\ref{eq:elle_two_combination}). Without loss of generality, let any pair of these random variables be separately indexed by 1 and 2, with no particular order. Moreover, abusing notation, let $\rho_{j}$ represent its own realization.

\begin{equation} \label{eq:elle_one_combination}
f_{1}(\rho_{1},\rho_{2})=\frac{\rho_{1}+\rho_{2}}{2}-\left|\frac{\rho_{1}-\rho_{2}}{2}\right|
\end{equation}

\begin{equation} \label{eq:elle_two_combination}
f_{2}(\rho_{1},\rho_{2})=\left(\left(\frac{\rho_{1}+\rho_{2}}{2}\right)^{2}-\left(\frac{\rho_{1}-\rho_{2}}{2}\right)^{2}\right)^{\frac{1}{2}}
\end{equation}

\newtheorem{proposition}{Proposition}

\begin{proposition}
\begin{gather*}
\shortintertext{\emph{If}}
f_{1}(\rho_{1},\rho_{2})=\frac{\rho_{1}+\rho_{2}}{2}-\left|\frac{\rho_{1}-\rho_{2}}{2}\right| \quad \emph{and} \quad g_{1}(\rho_{1},\rho_{2})=\min(\rho_{1},\rho_{2}) \\
\shortintertext{\emph{then}} f_{1}(\rho_{1},\rho_{2}) \equiv g_{1}(\rho_{1},\rho_{2}) \quad \forall \rho_{j} \emph{.}
\end{gather*}
\end{proposition}

\begin{proof}
We prove Proposition 1 by cases.
\\[7pt]
Case 1. Let $\rho_{1} < \rho_{2}$.
\begin{gather*}
g_{1}=\rho_{1} \\
f_{1}=\frac{\rho_{1}+\rho_{2}+\rho_{1}-\rho_{2}}{2}=\frac{2\rho_{1}}{2}=\rho_{1}
\end{gather*}
Case 2. Let $\rho_{1} > \rho_{2}$.
\begin{gather*}
g_{1}=\rho_{2} \\
f_{1}=\frac{\rho_{1}+\rho_{2}-\rho_{1}+\rho_{2}}{2}=\frac{2\rho_{2}}{2}=\rho_{2}
\end{gather*}
Case 3. If $\rho_{1} = \rho_{2}$ then the proposition is trivially valid.
\end{proof}

\begin{proposition}
\begin{gather*}
\shortintertext{\emph{If}}
f_{2}(\rho_{1},\rho_{2})=\left(\left(\frac{\rho_{1}+\rho_{2}}{2}\right)^{2}-\left(\frac{\rho_{1}-\rho_{2}}{2}\right)^{2}\right)^{\frac{1}{2}} \quad \emph{and} \quad g_{2}(\rho_{1},\rho_{2})=(\rho_{1}\rho_{2})^{\frac{1}{2}} \\
\shortintertext{\emph{then}} \quad f_{1}(\rho_{1},\rho_{2}) \equiv g_{1}(\rho_{1},\rho_{2}) \quad \forall \rho_{j} \emph{.}
\end{gather*}
\end{proposition}

\begin{proof}
We prove Proposition 2 by simply using arithmetic operations.
\\
\begin{gather*}
f_{2}^{2}=\frac{\rho_{1}^{2}+2\rho_{1}\rho_{2}+\rho_{2}^{2}-\rho_{1}^{2}+2\rho_{1}\rho_{2}-\rho_{2}^{2}}{4}=\frac{4\rho_{1}\rho_{2}}{4}=\rho_{1}\rho_{2}
\\
f_{2}=(\rho_{1}\rho_{2})^{\frac{1}{2}}
\end{gather*}
\end{proof}

By proving equivalence between $f_{1,2}$ and $g_{1,2}$, respectively, it is trivial to extend random vector component combinations to three or more random variables.

\begin{align}
\shortintertext{Scenario 1.}
f_{1}(\rho) \equiv g_{1}(\rho) & \equiv \min_{j}(\rho_{1},\rho_{2},\ldots,\rho_{m-1},\rho_{m})
\\[7pt]
\shortintertext{Scenario 2.}
f_{2}(\rho) \equiv g_{2}(\rho) & \equiv \left(\prod_{j}^{m}\rho_{j}\right)^{\frac{1}{m}}
\end{align}

\subsection*{Inventory of model parameter values}

S~Table~\ref{S1_Table} provides an overview of the model parameters designated as free values, along with their respective prior ranges. These ranges are based on data compiled from various literature sources and generally reflect informed estimates derived from analogous biological systems. However, we have chosen to adopt broad bounds for these ranges, relying on our parameter inference schemes to determine biophysically relevant values. This table is a replica of the inventory presented in our companion study \cite{ramirez-sierra_ai-powered_2024}, and it is included here for simplicity.

\begin{table}[ht!]
\begin{adjustwidth}{-1.75in}{0in}
\caption{{\bf Summary of inferred (free) model parameters.}}
\centering
\begin{tabular}{ | c | c | c | c | c | c | c | }
\hline
\rowcolor{lightgray}
Name & Alias & Description & \parbox[c][2.5em][c]{3.75em}{\centering Lower \\ Bound} & \parbox[c][2.5em][c]{3.75em}{\centering Upper \\ Bound} & Units & Note \\
\hline
$h_{\textrm{act},\textit{Nanog},\textrm{NANOG}}$ & \textit{Nanog}\_NANOG & \parbox[c][2.5em][c]{10em}{\centering Half-saturation level \\ (self-activation)} & 0 & 1000 & [pc] & \parbox[c][2.5em][c]{7.5em}{\centering Primary \\ Core GRN} \\
$h_{\textrm{act},\textit{Gata6},\textrm{GATA6}}$ & \textit{Gata6}\_GATA6 & \parbox[c][2.5em][c]{10em}{\centering Half-saturation level \\ (self-activation)} & 0 & 1000 & [pc] & \parbox[c][2.5em][c]{7.5em}{\centering Primary \\ Core GRN} \\
$h_{\textrm{rep},\textit{Gata6},\textrm{NANOG}}$ & \textit{Gata6}\_NANOG & \parbox[c][2.5em][c]{10em}{\centering Half-saturation level \\ (mutual-repression)} & 0 & 1000 & [pc] & \parbox[c][2.5em][c]{7.5em}{\centering Primary \\ Core GRN} \\
$h_{\textrm{rep},\textit{Nanog},\textrm{GATA6}}$ & \textit{Nanog}\_GATA6 & \parbox[c][2.5em][c]{10em}{\centering Half-saturation level \\ (mutual-repression)} & 0 & 1000 & [pc] & \parbox[c][2.5em][c]{7.5em}{\centering Primary \\ Core GRN} \\
\hline
$h_{\textrm{act},\textit{Fgf4},\textrm{NANOG}}$ & \textit{Fgf4}\_NANOG & \parbox[c][2.5em][c]{10em}{\centering Half-saturation level \\ (activation)} & 0 & 1000 & [pc] & \parbox[c][2.5em][c]{7.5em}{\centering Secondary \\ Core GRN} \\
$h_{\textrm{act},\textit{Gata6},\textrm{A-ERK}}$ & \textit{Gata6}\_A-ERK & \parbox[c][2.5em][c]{10em}{\centering Half-saturation level \\ (activation)} & 0 & 1000 & [pc] & \parbox[c][2.5em][c]{7.5em}{\centering Secondary \\ Core GRN} \\
$h_{\textrm{rep},\textit{Fgf4},\textrm{GATA6}}$ & \textit{Fgf4}\_GATA6 & \parbox[c][2.5em][c]{10em}{\centering Half-saturation level \\ (repression)} & 0 & 1000 & [pc] & \parbox[c][2.5em][c]{7.5em}{\centering Secondary \\ Core GRN} \\
$h_{\textrm{rep},\textit{Nanog},\textrm{A-ERK}}$ & \textit{Nanog}\_A-ERK & \parbox[c][2.5em][c]{10em}{\centering Half-saturation level \\ (repression)} & 0 & 1000 & [pc] & \parbox[c][2.5em][c]{7.5em}{\centering Secondary \\ Core GRN} \\
\hline
$\tau_{\textrm{escape}}$ & $k_{\textrm{escape}}^{-1}$ & \parbox[c][2.5em][c]{10em}{\centering Mean escape time \\ (FGF4)} & 300 & 4500 & [s] & \parbox[c][2.5em][c]{7.5em}{\centering Cell-Cell \\ Communication} \\
$\tau_{\textrm{exchange}}$ & $k_{\textrm{exchange}}^{-1}$ & \parbox[c][2.5em][c]{10em}{\centering Mean exchange time \\ (FGF4)} & 30 & 4200 & [s] & \parbox[c][2.5em][c]{7.5em}{\centering Cell-Cell \\ Communication} \\
$\chi_{\textrm{auto}}$ & $1-\chi_{\textrm{para}}$ & Autocrine signaling fraction & 0 & 1 &  & \parbox[c][2.5em][c]{7.5em}{\centering Cell-Cell \\ Communication} \\
\hline
$\tau_{\textrm{pho},\textrm{ERK}}$ & $k_{\textrm{pho},\textrm{A-ERK}}^{-1}$ & \parbox[c][2.5em][c]{10em}{\centering Half-turnover time \\ (phosphorylation)} & 300 & 43200 & [s] & \parbox[c][2.5em][c]{7.5em}{\centering Signaling \\ Pathway} \\
$\tau_{\textrm{doh},\textrm{ERK}}$ & $k_{\textrm{doh},\textrm{I-ERK}}^{-1}$ & \parbox[c][2.5em][c]{10em}{\centering Half-turnover time \\ (dephosphorylation)} & 30 & 43200 & [s] & \parbox[c][2.5em][c]{7.5em}{\centering Signaling \\ Pathway} \\
$\tau_{\textrm{pho},\textrm{NANOG}}$ & $k_{\textrm{pho},\textrm{P-NANOG}}^{-1}$ & \parbox[c][2.5em][c]{10em}{\centering Half-turnover time \\ (phosphorylation)} & 300 & 43200 & [s] & \parbox[c][2.5em][c]{7.5em}{\centering Signaling \\ Pathway} \\
$\tau_{\textrm{doh},\textrm{NANOG}}$ & $k_{\textrm{doh},\textrm{NANOG}}^{-1}$ & \parbox[c][2.5em][c]{10em}{\centering Half-turnover time \\ (dephosphorylation)} & 30 & 43200 & [s] & \parbox[c][2.5em][c]{7.5em}{\centering Signaling \\ Pathway} \\
\hline
\parbox[c][2.5em][c]{7.5em}{\centering Mean Initial \\ mRNA Count} & \parbox[c][2.5em][c]{7.5em}{\centering \textit{Nanog}\_\textit{Gata6} \\ mRNA} & Initial condition & 0 & 250 & [mc] &  \\
\parbox[c][2.5em][c]{7.5em}{\centering Mean Initial \\ PROTEIN Count} & \parbox[c][2.5em][c]{7.5em}{\centering NANOG\_GATA6 \\ PROTEIN} & Initial condition & 0 & 1000 & [pc] &  \\
$\tau_{\textrm{d},\textrm{FGF4}}$ & $k_{\textrm{p},\textrm{d},\textrm{FGF4}}^{-1}$ & Lifetime or half-life & 300 & 28800 & [s] & \parbox[c][2.5em][c]{7.5em}{\centering Molecular \\ Stability} \\
$\tau_{\textrm{d},\textrm{M-FGFR-FGF4}}$ & $k_{\textrm{p},\textrm{d},\textrm{M-FGFR-FGF4}}^{-1}$ & Lifetime or half-life & 300 & 28800 & [s] & \parbox[c][2.5em][c]{7.5em}{\centering Molecular \\ Stability} \\
\hline
\end{tabular}
\begin{flushleft}
Notation: act = activation; rep = repression; pho = phosphorylation; doh = dephosphorylation; [pc] = [protein copies]; [mc] = [mRNA copies]; [s] = [seconds].
\end{flushleft}
\label{S1_Table}
\end{adjustwidth}
\end{table}
\clearpage

\begin{adjustwidth}{-1.75in}{0in} 
\includegraphics[width = 7in, height = 7in]{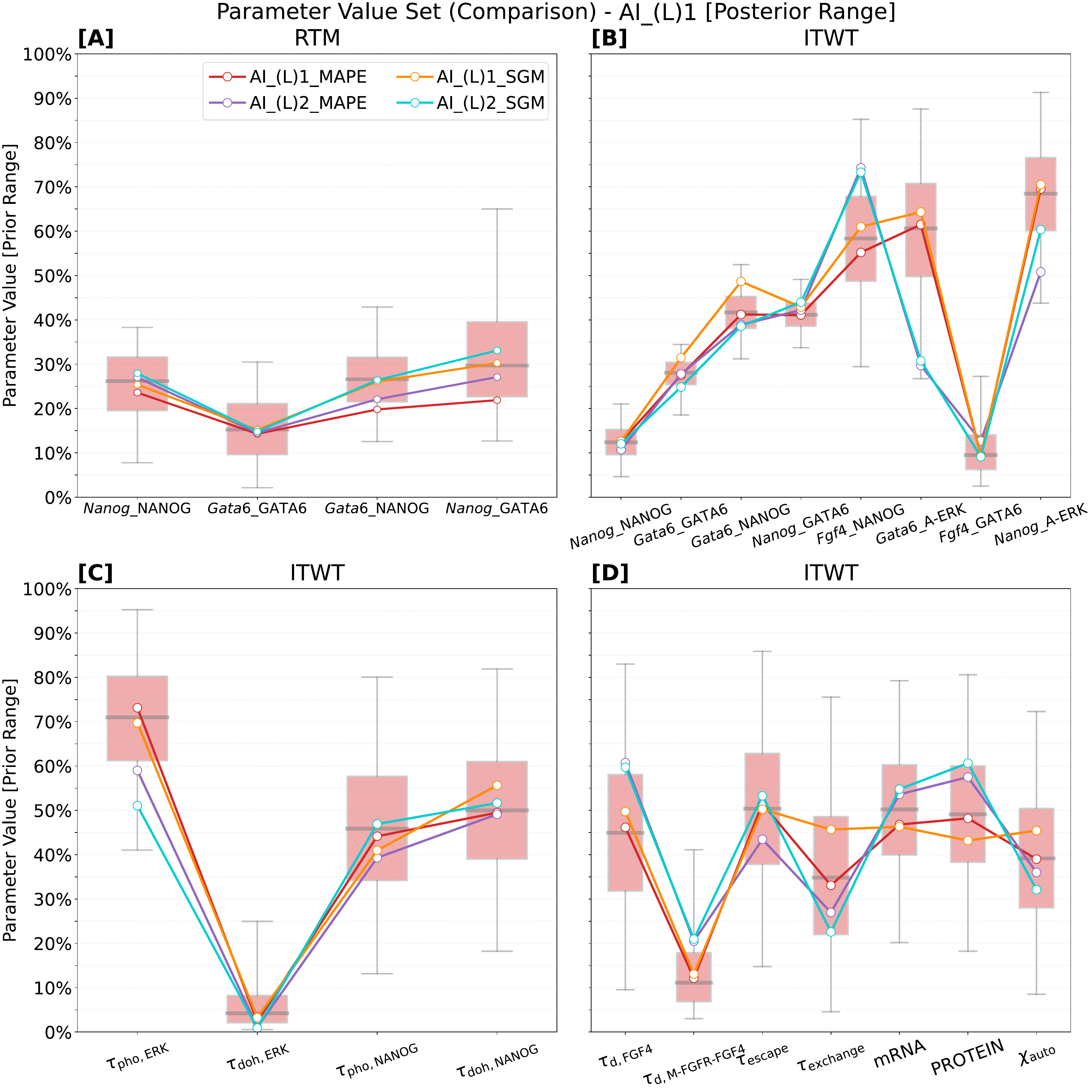} \centering 
\end{adjustwidth} 
\begin{figure}[hpt!]
\begin{adjustwidth}{-1.75in}{0in} 
\caption{{\bf Comparing MAPE versus SGM for SNPE-inspired workflow.} To illustrate the validity and practicality of the sample geometric median (SGM), we compute it for all the AI-MAPE workflow runs and show it alongside their associated maximum a posteriori probability estimates (MAPEs). {\bf [A]} RTM system: red and purple colors indicate MAPEs, orange and turquoise colors indicate SGMs, for AI\_1 and A1\_2 runs, respectively. {\bf [B-D]} ITWT system: same as [A], but for AI\_L1 and AI\_L2 runs, respectively. {\bf [A-D]} Note that despite some minor variations, SGMs and MAPEs move side by side, trending in parallel overall. We remark that, in general, the separate SGM component values (inferred model parameters) do not necessarily follow the univariate median of their related one-dimensional marginal posteriors; this observation is expected due to co-dependency among the model parameter values. For clarity, comparisons should proceed pairwise: red-orange, and purple-turquoise. See Fig~\ref{fig3}[C] and Fig~\ref{fig7} for additional details.}
\label{S1_Fig}
\end{adjustwidth} 
\end{figure}
\clearpage

\begin{adjustwidth}{-1.75in}{0in} 
\includegraphics[width = 5.5in, height = 5.5in]{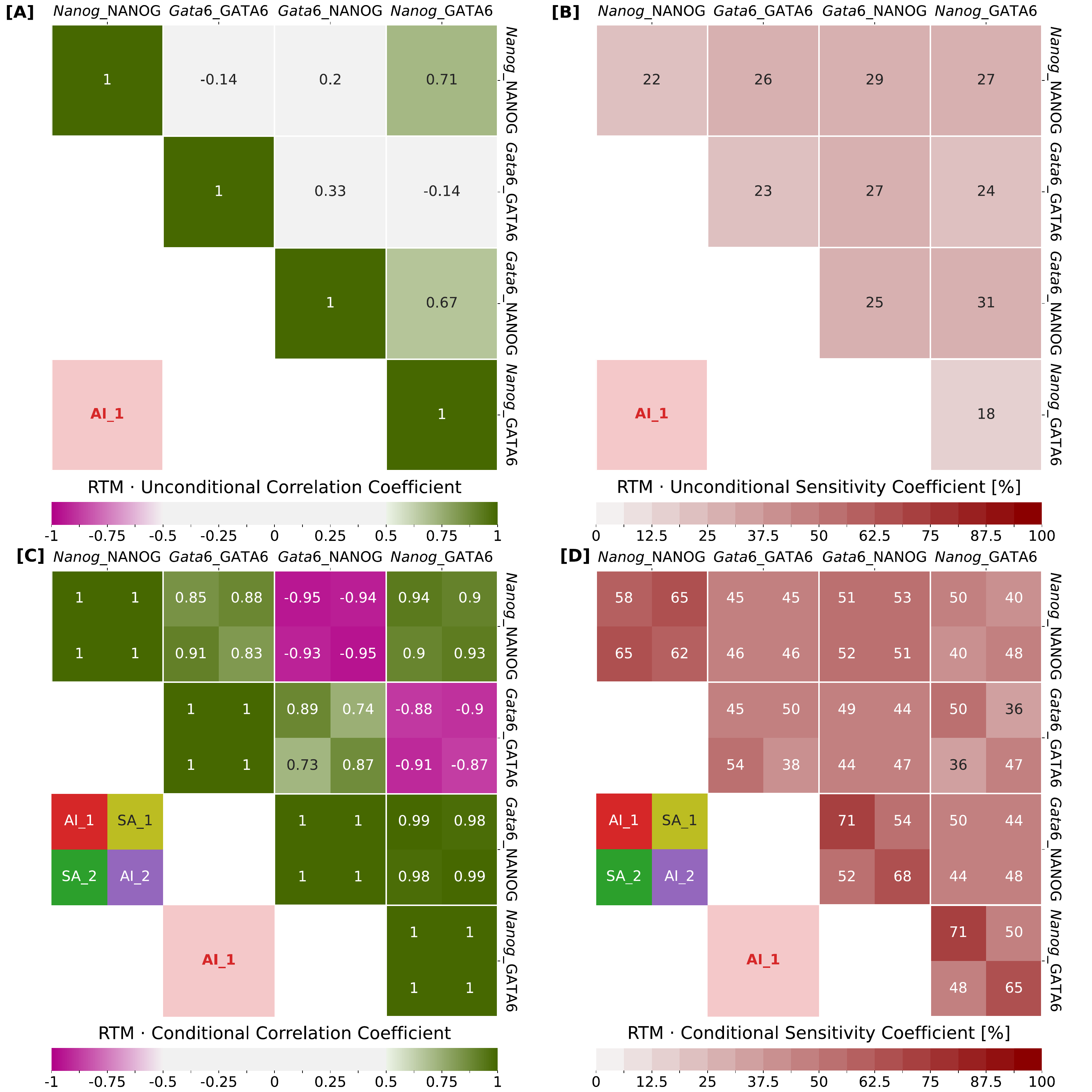} \centering 
\end{adjustwidth} 
\begin{figure}[hpt!]
\begin{adjustwidth}{-1.75in}{0in} 
\caption{{\bf RTM system. Unconditional model parameter correlation and sensitivity matrices.} {\bf [A]} Linear correlation coefficients directly extracted from AI\_1 posterior distribution (global reference). Note weak linear correlations among several parameter pairs (indicating no potential compensatory mechanisms between each other), as well as moderate linear correlations among \textit{Nanog}\_NANOG-\textit{Nanog}\_GATA6 and \textit{Gata6}\_NANOG-\textit{Nanog}\_GATA6. Notice also that this calculation is possible because the trained ANN acts as a surrogate of the simulator and directly approximates the AI\_1 posterior; i.e., no additional simulations are required to extract Pearson's correlation coefficients. {\bf [D]} Sensitivity to value changes for all parameters; unconditional AI\_1 posterior. Sensitivity equal to 0\% signifies that any value within its prior range can recover the ideal or target system behavior (while holding all the other parameters fixed). Sensitivity equal to 100\% signifies that the value must fall within a singular parameter bin (histograms were created with 250 bins per dimension). Note weak (0-25\%) and moderate (25-50\%) sensitivities for all parameter relations. No strong (50-100\%) sensitivities are recorded. {\bf [C, D]} Conditional model parameter correlation and sensitivity matrices. These two panels are only intended for quick reference; see Fig~\ref{fig4}[C, D] for complete details.}
\label{S2_Fig}
\end{adjustwidth} 
\end{figure}
\clearpage

\begin{adjustwidth}{-1.75in}{0in} 
\includegraphics[width = 5.5in, height = 5.5in]{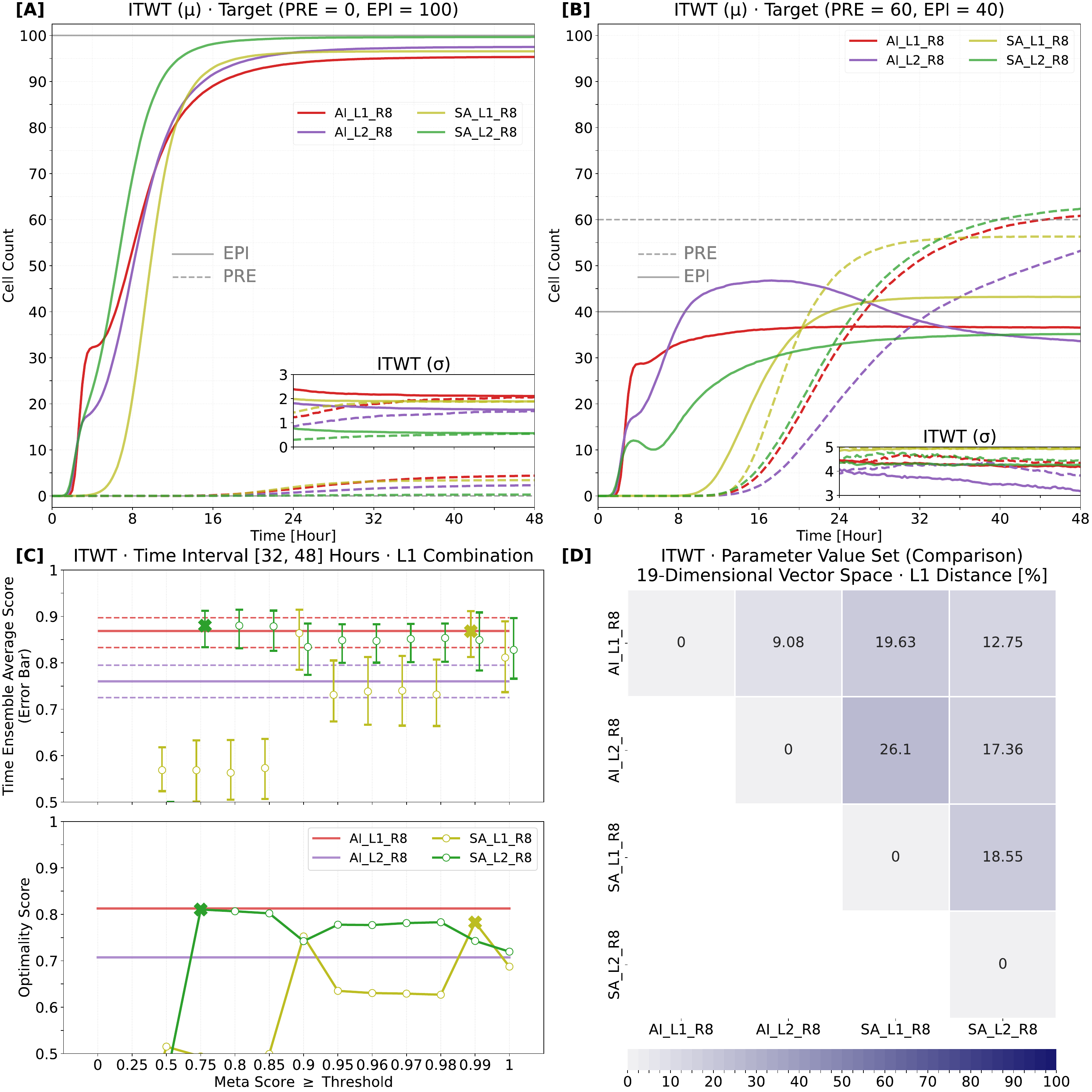} \centering 
\end{adjustwidth} 
\begin{figure}[hpt!]
\begin{adjustwidth}{-1.75in}{0in} 
\caption{{\bf Comparison at equal number of simulation-inference rounds (R8) of ITWT system features.} {\bf [A, B]} Behavior at tissue scale for all four parameter sets (100-cell grid). Correct cell-fate proportions should be reached and sustained within a time window between 32 and 48 hours for both system configurations separately. {\bf [A]} The target of 0 PRE cells and 100 EPI cells applies to the configuration without cell-cell communication (nonfunctional signaling). {\bf [B]} The target of 60 PRE cells and 40 EPI cells applies to the configuration with cell-cell communication (functional signaling). {\bf [C]} Selection process for the best or optimal parameter sets of the two SA runs (SA\_L1 and SA\_L2). Statistics of the selection process for the optimal parameter sets of the two AI runs (AI\_L1 and AI\_L2) are shown for reference. Two techniques are used to produce the joint configuration-score data from the two marginal configuration-score time series (L1-norm- and L2-norm-inspired combinations), but calculating TEAS and OS for all method-run pairs requires a common technique: L1 combination is used for simplicity. See Fig~\ref{fig3} caption for definitions of Time Ensemble Average Score ``TEAS'' and Optimality Score ``OS''. Crosses highlight the best parameter sets for the SA runs. Notice that naively picking the best possible meta-score (filtering threshold equal to 1) does not directly translate to finding the best actual performance. {\bf [D]} Distance matrix contrasting each pair of parameter sets: the parameter value sets are assumed to be elements of an abstract nineteen-dimensional vector space. Notice stronger agreement among AI\_L1, AI\_L2, and SA\_L2 than between SA\_L1 and SA\_L2. The L1 metric (normalized between 0\% and 100\%) was used to quantify distances between parameter vectors. See Fig~\ref{fig6} for reference.}
\label{S3_Fig}
\end{adjustwidth} 
\end{figure}
\clearpage

\begin{adjustwidth}{-1.75in}{0in} 
\includegraphics[width = 6.25in, height = 6.25in]{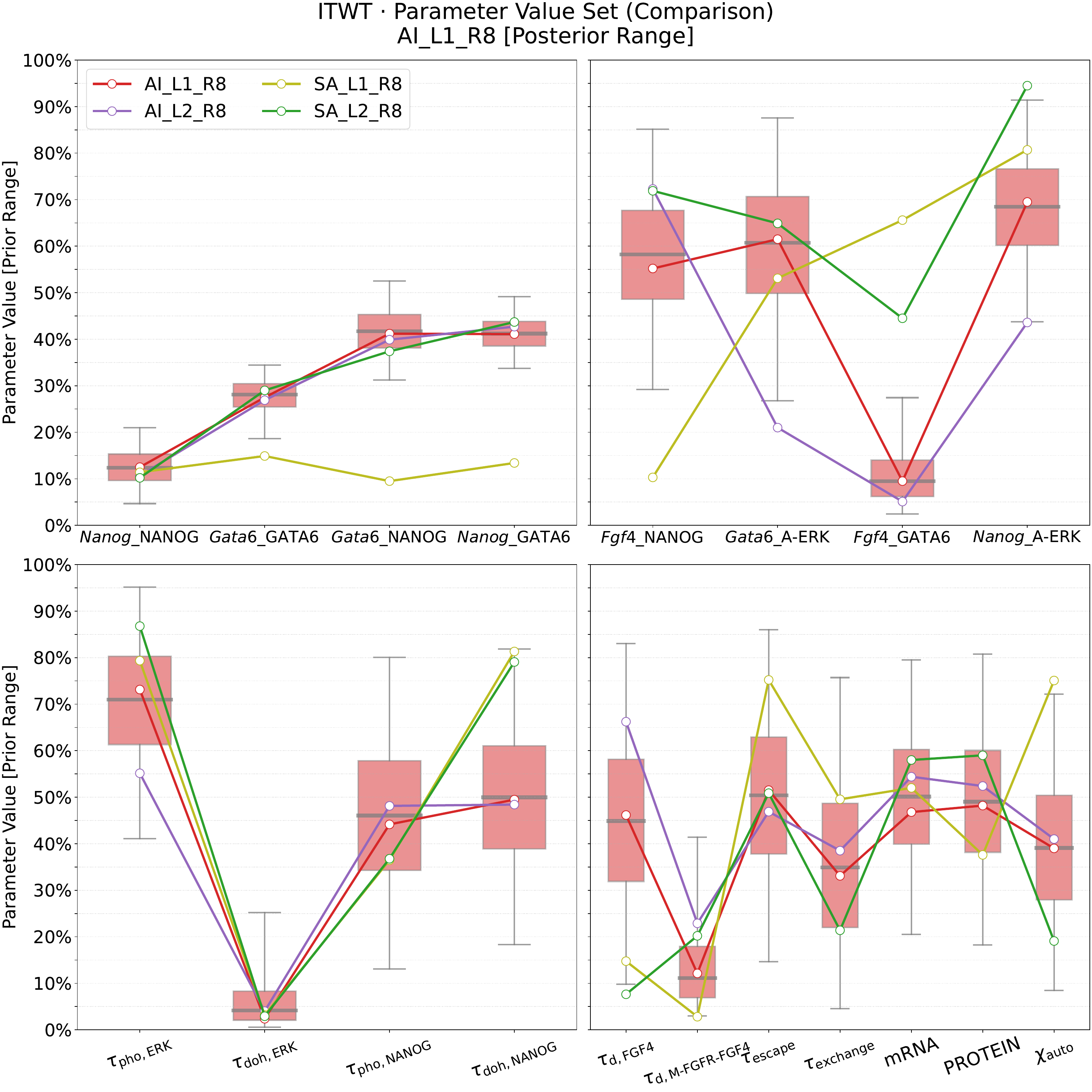} \centering 
\end{adjustwidth} 
\begin{figure}[hpt!]
\begin{adjustwidth}{-1.75in}{0in} 
\caption{{\bf ITWT system. Comparison at equal number of simulation-inference rounds (R8) of inferred/optimal parameter sets.} Box-and-whisker diagrams for each AI\_L1-related one-dimensional marginal posterior distribution are shown as baselines. See Fig~\ref{fig3}[C] caption for definitions of boxes and whiskers. Parameter values fall under normalized prior ranges. Top row: inferred core GRN interaction parameters (cell-scale dynamics). Bottom row: inferred spatial coupling (tissue-scale dynamics) and initial condition (mRNA-PROTEIN counts) parameters. Notice strong agreement among parameter sets AI\_L1, AI\_L2, and SA\_L2 for predictions of primary core GRN interaction values (top-left panel), but weak agreement for predictions of secondary core GRN interaction values (top-right panel). Note also moderate agreement among all parameter sets for predictions of other values. These differences indicate emergence of potential compensation mechanisms among parameters, highlighting distinct exploitable strategies to achieve the ideal or target system behavior. See Fig~\ref{fig7} for reference.}
\label{S4_Fig}
\end{adjustwidth} 
\end{figure}
\clearpage

\begin{adjustwidth}{-1.75in}{0in} 
\includegraphics[width = 7in, height = 7in]{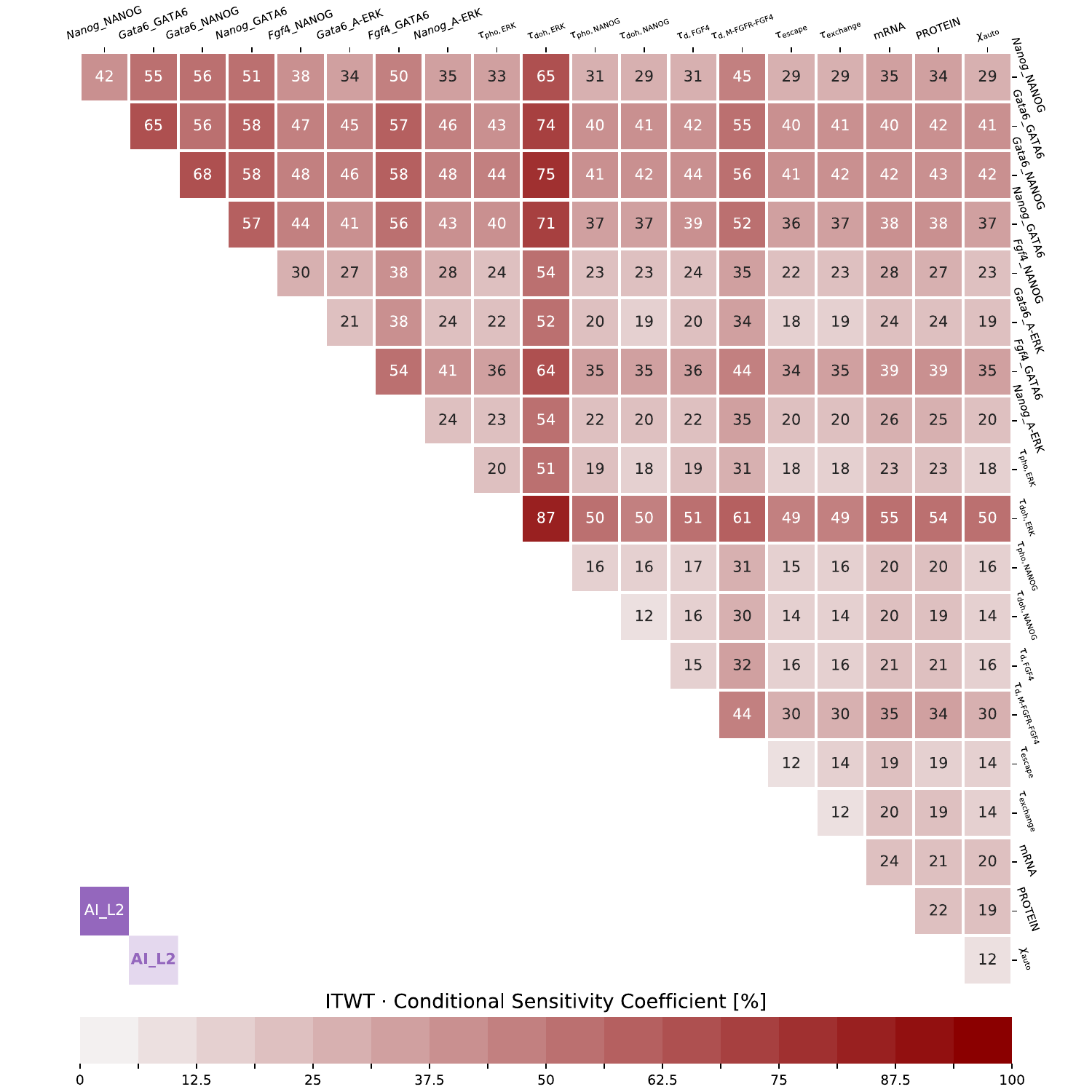} \centering 
\end{adjustwidth} 
\begin{figure}[hpt!]
\begin{adjustwidth}{-1.75in}{0in} 
\caption{{\bf ITWT system. Conditional model parameter sensitivity matrix.} Sensitivity to value changes for all parameters; AI\_L2 posterior conditional on its own MAPE. Altogether, this result is identical to the conclusion of Fig~\ref{fig9}. Note weak sensitivities (0-25\%) for (only diagonal entries) \textit{Gata6}\_A-ERK, \textit{Nanog}\_A-ERK, $\tau_{\textrm{pho},\textrm{ERK}}$, $\tau_{\textrm{pho},\textrm{NANOG}}$, $\tau_{\textrm{doh},\textrm{NANOG}}$, $\tau_{\textrm{d},\textrm{FGF4}}$, $\tau_{\textrm{escape}}$, $\tau_{\textrm{exchange}}$, and $\chi_{\textrm{auto}}$; moderate sensitivities (25-50\%) for \textit{Nanog}\_NANOG, \textit{Fgf4}\_NANOG, mRNA, and PROTEIN. Notice also strong sensitivities (50-100\%) for \textit{Gata6}\_GATA6, \textit{Gata6}\_NANOG, \textit{Nanog}\_GATA6, \textit{Fgf4}\_GATA6, $\tau_{\textrm{doh},\textrm{ERK}}$, and $\tau_{\textrm{d},\textrm{M-FGFR-FGF4}}$. Strong sensitivities reflect low tolerance to value fluctuations, given any other parameter is fixed at corresponding MAPE, for recapitulating the ideal or target system behavior. See Fig~\ref{fig4}[D] caption for definition of sensitivity.}
\label{S5_Fig}
\end{adjustwidth} 
\end{figure}
\clearpage

\begin{adjustwidth}{-1.75in}{0in} 
\includegraphics[width = 7in, height = 7in]{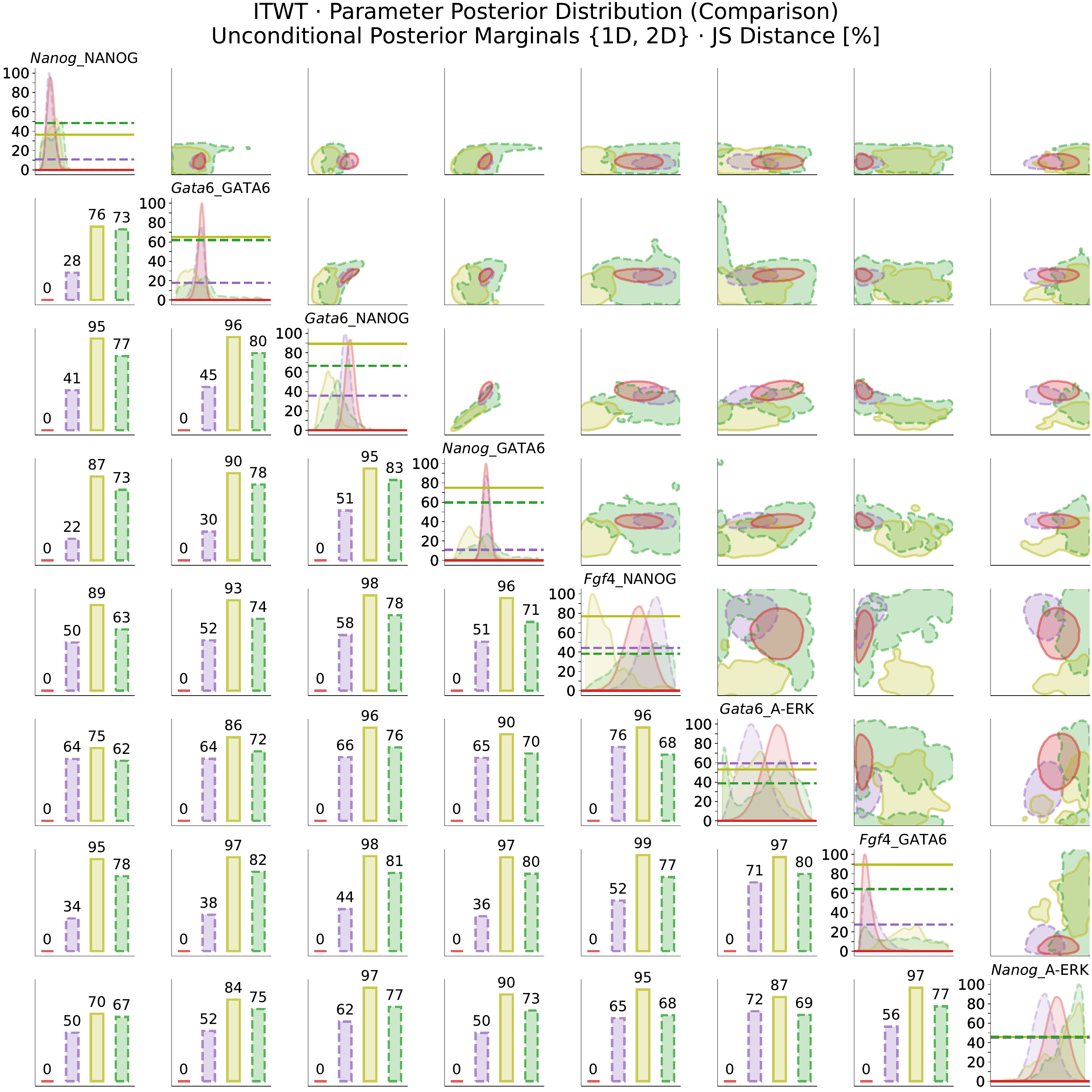} \centering 
\end{adjustwidth} 
\begin{figure}[hpt!]
\begin{adjustwidth}{-1.75in}{0in} 
\caption{{\bf ITWT system. Analysis of estimated marginalized posterior distributions for the (primary and secondary) core GRN motif relations.} Comparisons between one- and two-dimensional (diagonal and upper elements) posterior marginals. Distances with respect to reference case ``AI\_L1'' between one- and two-dimensional (diagonal and lower elements) posterior marginals are also shown. Distances are normalized between 0\% (minimally divergent histograms) and 100\% (maximally divergent histograms). Distances (shown as percentages) between raw probability vectors were quantified using the Jensen-Shannon metric (base 2). Note that off-diagonal entries (lower and upper sectors) symmetrically correspond to one another. For upper-triangular entries, we filter histogram regions where probability masses are greater than or equal to 0.25 (arbitrary threshold), and create smooth projections via Gaussian kernel density estimates. See Fig~\ref{fig10} for reference.}
\label{S6_Fig}
\end{adjustwidth} 
\end{figure}
\clearpage

\begin{adjustwidth}{-1.75in}{0in} 
\includegraphics[width = 7in, height = 7in]{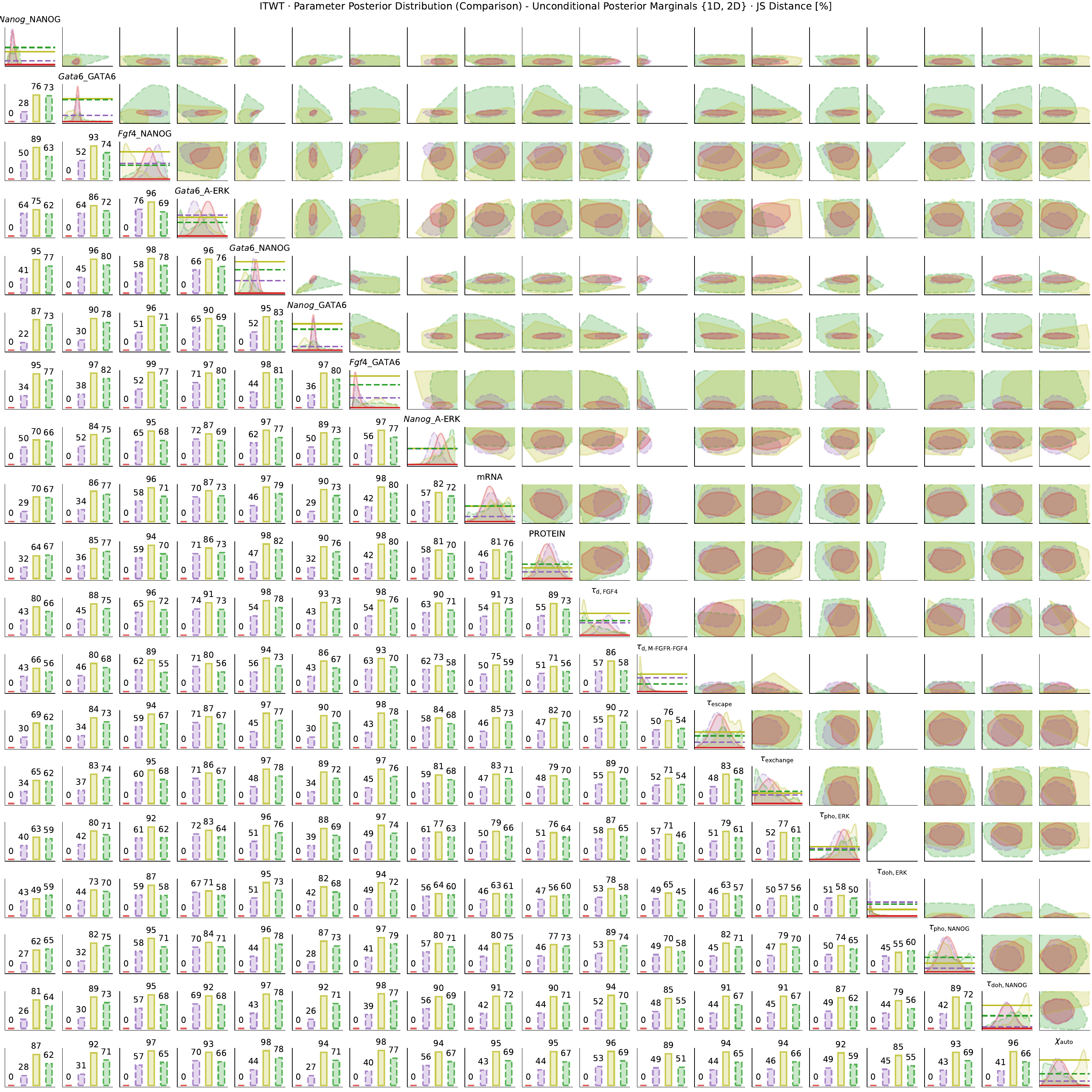} \centering 
\end{adjustwidth} 
\begin{figure}[hpt!]
\begin{adjustwidth}{-1.75in}{0in} 
\caption{{\bf ITWT system. Analysis of estimated marginalized posterior distributions: complete plot of parameter relations.} Comparisons between one- and two-dimensional (diagonal and upper elements) posterior marginals. Distances with respect to reference case ``AI\_L1'' between one- and two-dimensional (diagonal and lower elements) posterior marginals are also shown. Distances are normalized between 0\% (minimally divergent histograms) and 100\% (maximally divergent histograms). Distances (shown as percentages) between raw probability vectors were quantified using the Jensen-Shannon metric (base 2). Note that off-diagonal entries (lower and upper sectors) symmetrically correspond to one another. For upper-triangular entries, we filter histogram regions where probability masses are greater than or equal to 0.25 (arbitrary threshold), and create coarse projections via convex hull approximations. See Fig~\ref{fig10} for reference.}
\label{S7_Fig}
\end{adjustwidth} 
\end{figure}
\clearpage



\section*{Acknowledgments}

The successful completion of this research project is largely attributed to the collaboration and support of our esteemed colleagues and collaborators. We extend our gratitude to Roberto Covino and his lab for their invaluable guidance and expertise throughout the project, especially for directing us towards the powerful SBI framework. Their insightful feedback and constructive critiques have significantly enhanced the quality of our work. Additionally, we acknowledge the financial support from the CMMS project and FIAS, which has been instrumental for the advancement of this research. We also thank the Center for Scientific Computing (CSC) at Goethe University Frankfurt for providing access to the Goethe-HLR cluster.



\bibliography{Manuscript_Catalog.bib}

\begin{thebibliography}{10}

\bibitem{torregrosa_mechanistic_2021}
Torregrosa G, Garcia-Ojalvo J.
\newblock Mechanistic models of cell-fate transitions from single-cell data.
\newblock Current Opinion in Systems Biology. 2021;26:79--86.
\newblock doi:{10.1016/j.coisb.2021.04.004}.

\bibitem{mjolsness_prospects_2019}
Mjolsness E.
\newblock Prospects for {Declarative} {Mathematical} {Modeling} of {Complex} {Biological} {Systems}.
\newblock Bulletin of Mathematical Biology. 2019;81(8):3385--3420.
\newblock doi:{10.1007/s11538-019-00628-7}.

\bibitem{bonnaffoux_wasabi_2019}
Bonnaffoux A, Herbach U, Richard A, Guillemin A, Gonin-Giraud S, Gros PA, et~al.
\newblock {WASABI}: a dynamic iterative framework for gene regulatory network inference.
\newblock BMC Bioinformatics. 2019;20(1):220.
\newblock doi:{10.1186/s12859-019-2798-1}.

\bibitem{cang_multiscale_2021}
Cang Z, Wang Y, Wang Q, Cho KWY, Holmes W, Nie Q.
\newblock A multiscale model via single-cell transcriptomics reveals robust patterning mechanisms during early mammalian embryo development.
\newblock PLOS Computational Biology. 2021;17(3):e1008571.
\newblock doi:{10.1371/journal.pcbi.1008571}.

\bibitem{forsyth_iven_2021}
Forsyth JE, Al-Anbaki AH, Fuente Rdl, Modare N, Perez-Cortes D, Rivera I, et~al.
\newblock {IVEN}: {A} quantitative tool to describe {3D} cell position and neighbourhood reveals architectural changes in {FGF4}-treated preimplantation embryos.
\newblock PLOS Biology. 2021;19(7):e3001345.
\newblock doi:{10.1371/journal.pbio.3001345}.

\bibitem{jiang_identification_2022}
Jiang R, Singh P, Wrede F, Hellander A, Petzold L.
\newblock Identification of dynamic mass-action biochemical reaction networks using sparse {Bayesian} methods.
\newblock PLOS Computational Biology. 2022;18(1):e1009830.
\newblock doi:{10.1371/journal.pcbi.1009830}.

\bibitem{dirk_recognition_2023}
Dirk R, Fischer JL, Schardt S, Ankenbrand MJ, Fischer SC.
\newblock Recognition and reconstruction of cell differentiation patterns with deep learning.
\newblock PLOS Computational Biology. 2023;19(10):e1011582.
\newblock doi:{10.1371/journal.pcbi.1011582}.

\bibitem{alamoudi_fitmulticell_2023}
Alamoudi E, Schälte Y, Müller R, Starruß J, Bundgaard N, Graw F, et~al.
\newblock {FitMultiCell}: simulating and parameterizing computational models of multi-scale and multi-cellular processes.
\newblock Bioinformatics. 2023;39(11):btad674.
\newblock doi:{10.1093/bioinformatics/btad674}.

\bibitem{prescott_efficient_2024}
Prescott TP, Warne DJ, Baker RE.
\newblock Efficient multifidelity likelihood-free {Bayesian} inference with adaptive computational resource allocation.
\newblock Journal of Computational Physics. 2024;496:112577.
\newblock doi:{10.1016/j.jcp.2023.112577}.

\bibitem{verdier_simulation-based_2023}
Verdier H, Laurent F, Cassé A, Vestergaard CL, Specht CG, Masson JB.
\newblock Simulation-based inference for non-parametric statistical comparison of biomolecule dynamics.
\newblock PLOS Computational Biology. 2023;19(2):e1010088.
\newblock doi:{10.1371/journal.pcbi.1010088}.

\bibitem{wang_missing_2024}
Wang Z, Hasenauer J, Schälte Y.
\newblock Missing data in amortized simulation-based neural posterior estimation.
\newblock PLOS Computational Biology. 2024;20(6):e1012184.
\newblock doi:{10.1371/journal.pcbi.1012184}.

\bibitem{wang_massive_2019}
Wang S, Fan K, Luo N, Cao Y, Wu F, Zhang C, et~al.
\newblock Massive computational acceleration by using neural networks to emulate mechanism-based biological models.
\newblock Nature Communications. 2019;10(1):4354.
\newblock doi:{10.1038/s41467-019-12342-y}.

\bibitem{coulier_multiscale_2021}
Coulier A, Hellander S, Hellander A.
\newblock A multiscale compartment-based model of stochastic gene regulatory networks using hitting-time analysis.
\newblock The Journal of Chemical Physics. 2021;154(18):184105.
\newblock doi:{10.1063/5.0010764}.

\bibitem{sukys_approximating_2022}
Sukys A, Öcal K, Grima R.
\newblock Approximating solutions of the {Chemical} {Master} equation using neural networks.
\newblock iScience. 2022;25(9).
\newblock doi:{10.1016/j.isci.2022.105010}.

\bibitem{coulier_systematic_2022}
Coulier A, Singh P, Sturrock M, Hellander A.
\newblock Systematic comparison of modeling fidelity levels and parameter inference settings applied to negative feedback gene regulation.
\newblock PLOS Computational Biology. 2022;18(12):e1010683.
\newblock doi:{10.1371/journal.pcbi.1010683}.

\bibitem{raimundez_posterior_2023}
Raimúndez E, Fedders M, Hasenauer J.
\newblock Posterior marginalization accelerates {Bayesian} inference for dynamical models of biological processes.
\newblock iScience. 2023;26(11).
\newblock doi:{10.1016/j.isci.2023.108083}.

\bibitem{schnoerr_approximation_2017}
Schnoerr D, Sanguinetti G, Grima R.
\newblock Approximation and inference methods for stochastic biochemical kinetics—a tutorial review.
\newblock Journal of Physics A: Mathematical and Theoretical. 2017;50(9):093001.
\newblock doi:{10.1088/1751-8121/aa54d9}.

\bibitem{cranmer_frontier_2020}
Cranmer K, Brehmer J, Louppe G.
\newblock The frontier of simulation-based inference.
\newblock Proceedings of the National Academy of Sciences. 2020;117(48):30055--30062.
\newblock doi:{10.1073/pnas.1912789117}.

\bibitem{franzin_landscape-based_2023}
Franzin A, Stützle T.
\newblock A landscape-based analysis of fixed temperature and simulated annealing.
\newblock European Journal of Operational Research. 2023;304(2):395--410.
\newblock doi:{10.1016/j.ejor.2022.04.014}.

\bibitem{stillman_generative_2023}
Stillman NR, Mayor R.
\newblock Generative models of morphogenesis in developmental biology.
\newblock Seminars in Cell \& Developmental Biology. 2023;147:83--90.
\newblock doi:{10.1016/j.semcdb.2023.02.001}.

\bibitem{lagergren_biologically-informed_2020}
Lagergren JH, Nardini JT, Baker RE, Simpson MJ, Flores KB.
\newblock Biologically-informed neural networks guide mechanistic modeling from sparse experimental data.
\newblock PLOS Computational Biology. 2020;16(12):e1008462.
\newblock doi:{10.1371/journal.pcbi.1008462}.

\bibitem{perez_efficient_2022}
Perez SM, Sailem H, Baker RE.
\newblock Efficient {Bayesian} inference for mechanistic modelling with high-throughput data.
\newblock PLOS Computational Biology. 2022;18(6):e1010191.
\newblock doi:{10.1371/journal.pcbi.1010191}.

\bibitem{papamakarios_masked_2018}
Papamakarios G, Pavlakou T, Murray I. Masked {Autoregressive} {Flow} for {Density} {Estimation}; 2018.
\newblock Available from: \url{http://arxiv.org/abs/1705.07057}.

\bibitem{greenberg_automatic_2019}
Greenberg DS, Nonnenmacher M, Macke JH. Automatic {Posterior} {Transformation} for {Likelihood}-{Free} {Inference}; 2019.
\newblock Available from: \url{http://arxiv.org/abs/1905.07488}.

\bibitem{deistler_truncated_2022}
Deistler M, Goncalves PJ, Macke JH. Truncated proposals for scalable and hassle-free simulation-based inference; 2022.
\newblock Available from: \url{http://arxiv.org/abs/2210.04815}.

\bibitem{boelts_simulation-based_2023}
Boelts J, Harth P, Gao R, Udvary D, Yáñez F, Baum D, et~al.
\newblock Simulation-based inference for efficient identification of generative models in computational connectomics.
\newblock PLOS Computational Biology. 2023;19(9):e1011406.
\newblock doi:{10.1371/journal.pcbi.1011406}.

\bibitem{xiong_efficient_2023}
Xiong Y, Yang X, Zhang S, He Z. An efficient likelihood-free {Bayesian} inference method based on sequential neural posterior estimation; 2023.
\newblock Available from: \url{http://arxiv.org/abs/2311.12530}.

\bibitem{dirmeier_simulation-based_2024}
Dirmeier S, Albert C, Perez-Cruz F. Simulation-based inference using surjective sequential neural likelihood estimation; 2024.
\newblock Available from: \url{http://arxiv.org/abs/2308.01054}.

\bibitem{goncalves_training_2020}
Gonçalves PJ, Lueckmann JM, Deistler M, Nonnenmacher M, Öcal K, Bassetto G, et~al.
\newblock Training deep neural density estimators to identify mechanistic models of neural dynamics.
\newblock eLife. 2020;9:e56261.
\newblock doi:{10.7554/eLife.56261}.

\bibitem{kaiser_simulation-based_2023}
Kaiser J, Stock R, Müller E, Schemmel J, Schmitt S. Simulation-based {Inference} for {Model} {Parameterization} on {Analog} {Neuromorphic} {Hardware}; 2023.
\newblock Available from: \url{http://arxiv.org/abs/2303.16056}.

\bibitem{kolmus_tuning_2024}
Kolmus A, Janquart J, Baka T, van Laarhoven T, Broeck CVD, Heskes T. Tuning neural posterior estimation for gravitational wave inference; 2024.
\newblock Available from: \url{http://arxiv.org/abs/2403.02443}.

\bibitem{dingeldein_simulation-based_2023}
Dingeldein L, Cossio P, Covino R.
\newblock Simulation-based inference of single-molecule force spectroscopy.
\newblock Machine Learning: Science and Technology. 2023;4(2):025009.
\newblock doi:{10.1088/2632-2153/acc8b8}.

\bibitem{dingeldein_amortized_2024}
Dingeldein L, Silva-Sánchez D, Dimprima E, Grigorieff N, Covino R, Cossio P.
\newblock Amortized identification of biomolecular conformations in {Cryo}-{EM} using simulation-based inference.
\newblock Biophysical Journal. 2024;123(3):282a.
\newblock doi:{10.1016/j.bpj.2023.11.1758}.

\bibitem{ramirez-sierra_ai-powered_2024}
Ramirez-Sierra MA, Sokolowski TR. {AI}-powered simulation-based inference of a genuinely spatial-stochastic model of early mouse embryogenesis; 2024.
\newblock Available from: \url{http://arxiv.org/abs/2402.15330}.

\bibitem{tejero-cantero_sbi_2020}
Tejero-Cantero A, Boelts J, Deistler M, Lueckmann JM, Durkan C, Gonçalves PJ, et~al.
\newblock sbi: {A} toolkit for simulation-based inference.
\newblock Journal of Open Source Software. 2020;5(52):2505.
\newblock doi:{10.21105/joss.02505}.

\bibitem{gorecki_amortized_2023}
Gorecki M, Macke JH, Deistler M. Amortized {Bayesian} {Decision} {Making} for simulation-based models; 2023.
\newblock Available from: \url{http://arxiv.org/abs/2312.02674}.

\bibitem{lueckmann_benchmarking_2021}
Lueckmann JM, Boelts J, Greenberg DS, Gonçalves PJ, Macke JH. Benchmarking {Simulation}-{Based} {Inference}; 2021.
\newblock Available from: \url{http://arxiv.org/abs/2101.04653}.

\bibitem{minsker_geometric_2015}
Minsker S.
\newblock Geometric median and robust estimation in {Banach} spaces.
\newblock Bernoulli. 2015;21(4):2308--2335.
\newblock doi:{10.3150/14-BEJ645}.

\bibitem{minsker_robust_2016}
Minsker S, Srivastava S, Lin L, Dunson DB. Robust and {Scalable} {Bayes} via a {Median} of {Subset} {Posterior} {Measures}; 2016.
\newblock Available from: \url{http://arxiv.org/abs/1403.2660}.

\bibitem{invernizzi_skipping_2022}
Invernizzi M, Krämer A, Clementi C, Noé F.
\newblock Skipping the {Replica} {Exchange} {Ladder} with {Normalizing} {Flows}.
\newblock The Journal of Physical Chemistry Letters. 2022;13(50):11643--11649.
\newblock doi:{10.1021/acs.jpclett.2c03327}.

\bibitem{tolley_methods_2024}
Tolley N, Rodrigues PLC, Gramfort A, Jones SR.
\newblock Methods and considerations for estimating parameters in biophysically detailed neural models with simulation based inference.
\newblock PLOS Computational Biology. 2024;20(2):e1011108.
\newblock doi:{10.1371/journal.pcbi.1011108}.

\bibitem{sgro_intracellular_2015}
Sgro AE, Schwab DJ, Noorbakhsh J, Mestler T, Mehta P, Gregor T.
\newblock From intracellular signaling to population oscillations: bridging size‐ and time‐scales in collective behavior.
\newblock Molecular Systems Biology. 2015;11(1):779.
\newblock doi:{10.15252/msb.20145352}.

\bibitem{chowdhary_journey_2022}
Chowdhary S, Hadjantonakis AK.
\newblock Journey of the mouse primitive endoderm: from specification to maturation.
\newblock Philosophical Transactions of the Royal Society B: Biological Sciences. 2022;377(1865):20210252.
\newblock doi:{10.1098/rstb.2021.0252}.

\bibitem{bruckner_information_2024}
Brückner DB, Tkačik G.
\newblock Information content and optimization of self-organized developmental systems.
\newblock Proceedings of the National Academy of Sciences. 2024;121(23):e2322326121.
\newblock doi:{10.1073/pnas.2322326121}.

\bibitem{franzin_revisiting_2019}
Franzin A, Stützle T.
\newblock Revisiting simulated annealing: {A} component-based analysis.
\newblock Computers \& Operations Research. 2019;104:191--206.
\newblock doi:{10.1016/j.cor.2018.12.015}.

\bibitem{boelts_flexible_2022}
Boelts J, Lueckmann JM, Gao R, Macke JH.
\newblock Flexible and efficient simulation-based inference for models of decision-making.
\newblock eLife. 2022;11:e77220.
\newblock doi:{10.7554/eLife.77220}.

\bibitem{albert_simulated_2015}
Albert C, Kuensch HR, Scheidegger A.
\newblock A {Simulated} {Annealing} {Approach} to {Approximate} {Bayes} {Computations}.
\newblock Statistics and Computing. 2015;25(6):1217--1232.
\newblock doi:{10.1007/s11222-014-9507-8}.

\bibitem{frank_input-output_2013}
Frank SA.
\newblock Input-output relations in biological systems: measurement, information and the {Hill} equation.
\newblock Biology Direct. 2013;8(1):31.
\newblock doi:{10.1186/1745-6150-8-31}.

\bibitem{saiz_coordination_2020}
Saiz N, Hadjantonakis AK.
\newblock Coordination between patterning and morphogenesis ensures robustness during mouse development.
\newblock Philosophical Transactions of the Royal Society B. 2020;doi:{10.1098/rstb.2019.0562}.

\bibitem{plusa_common_2020}
Płusa B, Piliszek A.
\newblock Common principles of early mammalian embryo self-organisation.
\newblock Development. 2020;147(dev183079).
\newblock doi:{10.1242/dev.183079}.

\bibitem{fange_stochastic_2010}
Fange D, Berg OG, Sjöberg P, Elf J.
\newblock Stochastic reaction-diffusion kinetics in the microscopic limit.
\newblock Proceedings of the National Academy of Sciences. 2010;107(46):19820--19825.
\newblock doi:{10.1073/pnas.1006565107}.

\bibitem{erban_stochastic_2020}
Erban R, Chapman SJ.
\newblock Stochastic {Modelling} of {Reaction}–{Diffusion} {Processes}.
\newblock Cambridge {Texts} in {Applied} {Mathematics}. Cambridge: Cambridge University Press; 2020.
\newblock Available from: \url{https://www.cambridge.org/core/books/stochastic-modelling-of-reactiondiffusion-processes/9BB8B46DE0B898FC019AFBEA95608FAE}.

\bibitem{barrows_parameter_2023}
Barrows D, Ilie S.
\newblock Parameter estimation for the reaction–diffusion master equation.
\newblock AIP Advances. 2023;13(6):065318.
\newblock doi:{10.1063/5.0150292}.

\bibitem{kang_multiscale_2019}
Kang HW, Erban R.
\newblock Multiscale {Stochastic} {Reaction}–{Diffusion} {Algorithms} {Combining} {Markov} {Chain} {Models} with {Stochastic} {Partial} {Differential} {Equations}.
\newblock Bulletin of Mathematical Biology. 2019;81(8):3185--3213.
\newblock doi:{10.1007/s11538-019-00613-0}.

\bibitem{varolgunes_interpretable_2020}
Varolgüneş YB, Bereau T, Rudzinski JF.
\newblock Interpretable embeddings from molecular simulations using {Gaussian} mixture variational autoencoders.
\newblock Machine Learning: Science and Technology. 2020;1(1):015012.
\newblock doi:{10.1088/2632-2153/ab80b7}.

\bibitem{allegre_nanog_2022}
Allègre N, Chauveau S, Dennis C, Renaud Y, Meistermann D, Estrella LV, et~al.
\newblock {NANOG} initiates epiblast fate through the coordination of pluripotency genes expression.
\newblock Nature Communications. 2022;13(1):3550.
\newblock doi:{10.1038/s41467-022-30858-8}.

\bibitem{bessonnard_icm_2017}
Bessonnard S, Coqueran S, Vandormael-Pournin S, Dufour A, Artus J, Cohen-Tannoudji M.
\newblock {ICM} conversion to epiblast by {FGF}/{ERK} inhibition is limited in time and requires transcription and protein degradation.
\newblock Scientific Reports. 2017;7(1):12285.
\newblock doi:{10.1038/s41598-017-12120-0}.

\bibitem{papamakarios_sequential_2019}
Papamakarios G, Sterratt DC, Murray I. Sequential {Neural} {Likelihood}: {Fast} {Likelihood}-free {Inference} with {Autoregressive} {Flows}; 2019.
\newblock Available from: \url{http://arxiv.org/abs/1805.07226}.

\bibitem{miller_contrastive_2023}
Miller BK, Weniger C, Forré P. Contrastive {Neural} {Ratio} {Estimation}; 2023.
\newblock Available from: \url{http://arxiv.org/abs/2210.06170}.

\bibitem{minsker_geometric_2023}
Minsker S, Strawn N. The {Geometric} {Median} and {Applications} to {Robust} {Mean} {Estimation}; 2023.
\newblock Available from: \url{http://arxiv.org/abs/2307.03111}.

\bibitem{arumugam_bayesian_2023}
Arumugam D, Ho MK, Goodman ND, Van~Roy B. Bayesian {Reinforcement} {Learning} with {Limited} {Cognitive} {Load}; 2023.
\newblock Available from: \url{http://arxiv.org/abs/2305.03263}.

\bibitem{kochenderfer_algorithms_2019}
Kochenderfer MJ, Wheeler TA.
\newblock Algorithms for {Optimization}.
\newblock The MIT Press; 2019.
\newblock Available from: \url{https://mitpress.mit.edu/9780262039420/algorithms-for-optimization/}.

\bibitem{glockler_variational_2022}
Glöckler M, Deistler M, Macke JH. Variational methods for simulation-based inference; 2022.
\newblock Available from: \url{http://arxiv.org/abs/2203.04176}.

\bibitem{marder_multiple_2011}
Marder E, Taylor AL.
\newblock Multiple models to capture the variability in biological neurons and networks.
\newblock Nature Neuroscience. 2011;14(2):133--138.
\newblock doi:{10.1038/nn.2735}.

\bibitem{massonis_distilling_2023}
Massonis G, Villaverde AF, Banga JR. Distilling identifiable and interpretable dynamic models from biological data; 2023.
\newblock Available from: \url{https://www.biorxiv.org/content/10.1101/2023.03.13.532340v2}.

\bibitem{lin_divergence_1991}
Lin J.
\newblock Divergence measures based on the {Shannon} entropy.
\newblock IEEE Transactions on Information Theory. 1991;37(1):145--151.
\newblock doi:{10.1109/18.61115}.

\bibitem{saiz_growth-factor-mediated_2020}
Saiz N, Mora-Bitria L, Rahman S, George H, Herder JP, Garcia-Ojalvo J, et~al.
\newblock Growth-factor-mediated coupling between lineage size and cell fate choice underlies robustness of mammalian development.
\newblock eLife. 2020;9.
\newblock doi:{10.7554/eLife.56079}.

\bibitem{raina_cell-cell_2021}
Raina D, Bahadori A, Stanoev A, Protzek M, Koseska A, Schröter C.
\newblock Cell-cell communication through {FGF4} generates and maintains robust proportions of differentiated cell types in embryonic stem cells.
\newblock Development. 2021;148(21):dev199926.
\newblock doi:{10.1242/dev.199926}.

\bibitem{fischer_salt-and-pepper_2023}
Fischer SC, Schardt S, Lilao-Garzón J, Muñoz-Descalzo S.
\newblock The salt-and-pepper pattern in mouse blastocysts is compatible with signaling beyond the nearest neighbors.
\newblock iScience. 2023;26(11).
\newblock doi:{10.1016/j.isci.2023.108106}.

\bibitem{ryan_lumen_2019}
Ryan AQ, Chan CJ, Graner F, Hiiragi T.
\newblock Lumen {Expansion} {Facilitates} {Epiblast}-{Primitive} {Endoderm} {Fate} {Specification} during {Mouse} {Blastocyst} {Formation}.
\newblock Developmental Cell. 2019;51(6):684--697.e4.
\newblock doi:{10.1016/j.devcel.2019.10.011}.

\bibitem{shahbazi_mechanisms_2020}
Shahbazi MN.
\newblock Mechanisms of human embryo development: from cell fate to tissue shape and back.
\newblock Development. 2020;147(14):dev190629.
\newblock doi:{10.1242/dev.190629}.

\end{thebibliography}




\end{document}